\documentclass[a4paper, amsfonts, amssymb, amsmath, reprint, showkeys, nofootinbib, twoside, onecolumn,notitlepage,,longbibliography]{revtex4-2}

\def\prd{{Phys.~Rev.~D}}        
\def\prl{{Phys.~Rev.~Lett.}}    
\usepackage{amsmath,amstext}
\usepackage[T1]{fontenc}
\usepackage{amssymb}

\usepackage{graphicx}
\usepackage{ae,aecompl}
\usepackage{combelow}
\usepackage{upgreek}
\usepackage{standalone}
\usepackage{hyperref}
\usepackage{amsmath}
\usepackage{amssymb}
\usepackage{mathtools}
\usepackage{bm}
\usepackage{cleveref}
\usepackage{tensor}
\usepackage{braket}
\usepackage{enumitem}
\usepackage{mhchem}
\usepackage{cancel}
\usepackage{tikz}
\usepackage{pgfplots}
\usetikzlibrary{external}
\usepackage{mathrsfs}
\usepackage{amsthm}

\usepackage{color}

\newcommand{\spc}{\quad \quad \quad}

\def\be{\begin{equation}}
\def\ee{\end{equation}}
\def\beq{\begin{eqnarray}}
\def\eeq{\end{eqnarray}}

\theoremstyle{definition}

\theoremstyle{theorem}
\newtheorem{theorem}{Theorem}

\theoremstyle{corollary}

\begin{document}
\title{Gapless non-hydrodynamic modes in relativistic kinetic theory}
\author{L.~Gavassino}
\email{lorenzo.gavassino@vanderbilt.edu}
\affiliation{Department of Mathematics, Vanderbilt University, Nashville, TN 37211, USA}

\begin{abstract}
We provide rigorous criteria to determine whether the non-hydrodynamic sector of a relativistic kinetic theory is gapless. These general criteria apply to both Boltzmann's equation and approximate models thereof, provided that the latter are consistent with thermodynamic principles. As an application of these criteria, we prove that the non-hydrodynamic sector is gapless whenever the scattering cross-section decays to zero at large energies. Physically, this happens because, if particles with very high energy (compared to the temperature) are free streaming, we can use them to build hot non-hydrodynamic waves, which live longer than any hydrodynamic wave. Since many standard cross-sections in quantum field theory vanish at high energies, the existence of these non-thermal long-lived waves is a rather general feature of relativistic systems. 
\end{abstract}

\maketitle

\noindent \textit{Introduction -} The existence of hydrodynamics usually relies on a ``separation of timescales'' assumption \cite{Kadanoff1963,huang_book,LindblomRelaxation1996,Hohenberg1977,Denicol2012Boltzmann,Hydro+2018,GavassinoFronntiers2021,Rocha:2023ilf}, which can be summarised as follows: Non-conserved quantities equilibrate (through collisions) over microscopic timescales, while conserved quantities equilibrate (through transport) over macroscopic timescales \cite{Glorioso2018}. To appreciate the rationale of this assumption, consider the following example. Let $\rho(t,x^j)$ be the coarse-grained energy density of a gas of photons immersed in a medium. Suppose that interactions with the medium happen through elastic scattering, so $\rho$ is a conserved density, and we can write a continuity equation $\partial_t \rho +\partial_j F^j=0$, with $F^j$ the energy flux \cite{Levermore1984}. If scatterings are frequent (compared to the evolution timescale of $\rho$), we can average the motion of photons over a few mean free times, and treat it as a random walk. Then, $F^j$ is uniquely determined in terms of $\rho$ by Fick's law, $F^j=-D\partial^j \rho$ \cite{mihalas_book}, resulting in a hydrodynamic equation for the density: $\partial_t \rho=D \partial_j \partial^j \rho$. If, instead, the interactions are rare, we cannot construct a differential equation in terms of $\rho$ alone, because $F^j$ depends on infinitely many microscopic degrees of freedom (namely, the initial direction of propagation of each individual photon). Thus, a self-contained macroscopic description of the flow does not exist.

The tendency of gases and liquids to exhibit a scale separation at small gradients is often taken for granted \cite{landau6,rezzolla_book,andersson2007review,Romatschke2017}. The standard argument goes as follows: Let $\tau_{\text{hy}}$ be the decay timescale of a hydrodynamic wave of interest (e.g. a soundwave), and let $\tau_{\text{non-hy}}$ be the equilibration time of the longest-living collective excitation that is not described by hydrodynamics (i.e. that is not protected from decaying by conservation laws). If
\begin{equation}\label{theassumption}
    \tau_{\text{non-hy}} \ll \tau_{\text{hy}} \, ,
\end{equation}
then all the ``non-hydrodynamic excitations'' decay faster than the hydrodynamic wave, which will dominate the late-time behavior. But since the damping factor of hydrodynamic waves due to viscosity scales like $\sim e^{-k^2 t}$, where $k$ is the wave number, one can make $\tau_{\text{hy}}\sim k^{-2}$ arbitrarily large by taking $k\rightarrow 0$. Thus, at sufficiently small gradients (i.e. at sufficiently large lengthscales), $\tau_{\text{hy}}$ will outlive all non-hydrodynamic excitations, and  \eqref{theassumption} must hold.

Unfortunately, the above argument has a flaw \cite{KovtunStickiness2011,kovtun_lectures_2012}. At infinitely large lengthscales, a fluid cell possesses infinitely many microscopic degrees of freedom. Thus, there may be an infinite sequence of collective excitations, whose equilibration times $\{\tau_n\}_{n\in \mathbb{N}}$ diverge at large $n$. Then, $\tau_{\text{non-hy}}=\infty$, and \eqref{theassumption} never occurs. For example, suppose that $\tau_n=n$. Then, if we assign to each excitation an initial amplitude $1/n^2$, the total perturbation decays like
\begin{equation}\label{fluctuzzuz}
    \sum_{n=1}^{\infty} \dfrac{e^{-t/n}}{n^2} \, \, \, \stackrel{\text{large } t}{\approx} \, \, \, \int_0^{+\infty} \dfrac{e^{-t/x}}{x^2} dx = \dfrac{1}{t} \, ,
\end{equation}
which survives longer than any soundwave (at non-zero viscosity). This possibility was extensively studied for non-relativistic gases \cite{Grad1963,McLennan1965,Caflish1980,CatlfishII1980}. There, it was shown that, if the collision cross-section decays to zero at high energies\footnote{In the mathematics literature, interactions of this type are called ``soft potentials'' \cite{Caflish1980,Strain2010}.}, then there is a continuous infinity of non-hydrodynamic excitations whose equilibration times are unbounded above. It was also shown that continuous superpositions of these excitations decay like ${\sim} e^{-t^b}$ (with $b{<}1$), outliving soundwaves. Whether realistic QFT interactions give rise to similar phenomena in relativistic systems is an open problem, known as the ``poles or cut?'' dilemma \cite{Moore:2018mma} (poles: $\tau_{\text{non-hy}}{<}\infty$; cut: $\tau_{\text{non-hy}}{=}\infty$). Several hints point toward ``cut'', coming from qualitative models \cite{Kurkela:2017xis,Gavassino2024InfiniteOrder,RochaGavassinoFluctu:2024afv} and numerical experiments \cite{Moore:2018mma,Ochsenfeld:2023wxz,RochaCut2024cge}. Here, we provide, for the first time, a kinetic-theory discussion that is both mathematically rigorous and physically intuitive, as it relies solely on general principles and quick estimates. We find that, in most cases, the answer is indeed ``cut''. From a mathematician's perspective, this is related to the fact that most QFT interactions happen to be ``soft'', in the sense of \cite{Dudynski1988} (see also \cite{Hu:2024tnn}).


The metric signature $(-,+,+,+)$ is adopted throughout, and we work in natural units: $c=\hbar=k_B=1$.

\newpage

\noindent \textit{Intuitive argument -} Consider a relativistic gas of particles with mass $m$, density $J^0$, and temperature $T$. Suppose that one particle in this gas has energy much higher than the thermal energy, i.e. $p^0 \gg \sqrt{m^2{+}T^2} \equiv E_{\text{th}}$. The probability that this particle travels freely for a time interval $t$ without interacting is $\mathcal{P}_s(t)= e^{-J^0 \upvarsigma t}$ \cite{Lechner:2018bzx,TavernierBook}, where $\upvarsigma$ is the mean collision cross-section. We take, as a ``realistic'' (i.e. QFT-inspired) high-energy interaction, the cross-section of $\lambda \varphi^4$ scalar field theory, $\upvarsigma \sim g/E_{\text{CM}}^2$, where $g$ is a constant, and $E_{\text{CM}}$ is the energy in the center of momentum frame \cite{Peskin_book}. For collisions between our high-energy particle and the rest of the gas, we have on average $E_{\text{CM}}^2 \approx 4E_{\text{th}} p^0$, giving \cite{Denicol:2022bsq}
\begin{equation}\label{jet}
\mathcal{P}_s(t)\approx \exp\bigg[{-}\dfrac{J^0 gt}{4E_{\text{th}} p^0}\bigg] \, .  
\end{equation}
Now, suppose that, instead of one high-energy particle, we have a population of such particles. If these are sufficiently diluted, we can treat each particle independently. Thus, working in the thermodynamic limit, imagine that there are $N$ particles with energy $p^0 (\gg E_{\text{th}})$, plus other $N/2^3$ particles with energy $2p^0$, plus other $N/3^3$ particles with energy $3p^0$, plus other $N/4^3$ particles with energy $4p^0$..., and so on (up to infinity). Then, the total energy associated with hard particles that have traveled freely without interacting from time 0 to time $t$ is a series, whose sum is finite:
\begin{equation}\label{poWeer}
\begin{split}
 E_{\text{non-hy}}(0) ={}& \pi^2 Np^0/6 \, , \\
E_{\text{non-hy}}(t) ={}& \sum_{n=1}^{\infty} \dfrac{Np^0}{n^2} \exp\bigg[{-}\dfrac{J^0 gt}{4E_{\text{th}} np^0} \bigg]  \, \, \, \stackrel{\text{large } t}{\approx} \, \, \, \dfrac{4NE_{\text{th}}(p^0)^2}{J^0gt} \, .\\
\end{split}
\end{equation}
Compare this situation with the photon gas in the introduction: $E_{\text{non-hy}}$ is the energy carried by particles that have traveled freely across the medium \textit{without} relaxing to a hydrodynamic constitutive relation. No macroscopic law relates the large-scale flux of $E_{\text{non-hy}}$ to large-scale densities. The amount of this ``non-hydrodynamic energy'' falls like $1/t$, while hydrodynamic waves decay exponentially. Therefore, free-streaming hard particles dominate the late-time transport, and $\tau_{\text{non-hy}}=\infty$. Similar conclusions hold in any gas where $\upvarsigma {\rightarrow} 0$ at large $E_{\text{CM}}$.

In summary: If the high-energy tails of the kinetic distribution are out of equilibrium, 
and the cross-section decays to zero at large energies (as in models of QCD thermalization \cite{Dusling:2009df,Dusling:2011fd,Kurkela:2017xis,Rocha:2021zcw}), high-energy particles form a non-hydrodynamic excitation that equilibrates slower than soundwaves, and $\tau_{\text{non-hy}}{=}\infty$. Note that also a small (e.g. an exponential) non-equilibrium tail lives longer than a sound wave (see Supplementary Material for two explicit examples).

\noindent \textit{The non-hydrodynamic sector is gapless -} Let $G(t{-}t')$ be the retarded linear-response Green function of the shear stress component $\uppi_{12}(t)$ induced by some external force $F(t')$, in the spatially homogenous limit. The singularities of its Fourier transform, $G(\omega)$, are the eigenfrequencies of the system, which (in kinetic theory)  lay on the imaginary axis of the complex $\omega$-plane \cite{Denicol_Relaxation_2011,GavassinoNonHydro2022}. Thus, by deforming the integration path as in figure \ref{fig:BranchCut}, we obtain \cite{WagnerGavassino2023jgq}
\begin{figure}
\begin{center}
\includegraphics[width=0.48\textwidth]{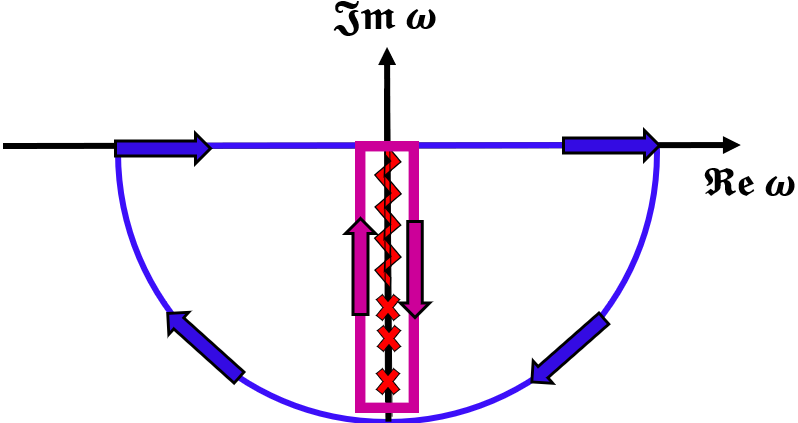}
	\caption{To evaluate the integral $G(t)=(2\pi)^{-1}\int_{\mathbb{R}} G(\omega)e^{-i\omega t}d\omega$ for $t>0$, we close the path for negative $\mathfrak{Im} \, \omega$ (blue line), since $e^{-i\omega t}$ decays exponentially there. Then, we shrink the path around the singularities (purple line). In the limit in which the vertical lines approach the imaginary axis, we obtain \eqref{sgagnatchz}, with $z(\nu)= (2\pi i)^{-1} \big[G(0^+{-}i\nu)-G(0^-{-}i\nu) \big]$ .}
	\label{fig:BranchCut}
	\end{center}
\end{figure}
\begin{equation}\label{sgagnatchz}
    G(t)=\Theta(t)\int_0^{+\infty} z(\nu) e^{-\nu t} d\nu\, ,
\end{equation}
where $z(\nu)$ is the spectral density of singularities at $\omega =-i\nu$. Comparing \eqref{jet} with \eqref{sgagnatchz} as $\nu \rightarrow 0$, we see that $z(\nu)$ is proportional to the density of hard particles per unit $(p^0)^{-1}\propto \nu$ that are disturbed by $F$. If the force $F$ couples with particles at all energies (this is the case of deformation rates \cite[Eq. (48)]{Denicol_Relaxation_2011}), $z(\nu)$ is a continuous function near $\nu{=}0$ \cite{Kurkela:2017xis,Ochsenfeld:2023wxz} (and not a sum of Dirac deltas), meaning that $G(\omega)$ has a branch cut touching the origin. Hence, the $\lambda \varphi^4$ gas has infinite non-hydrodynamic modes accumulating at $\omega{=}0$, making the non-hydrodynamic sector gapless \cite{Moore:2018mma,RochaCut2024cge}.

\newpage

In the rest of this Letter, we provide rigorous theorems supporting the above arguments, for general interactions. 

\noindent \textit{Formulation of the problem -} We consider a relativistic gas of classical\footnote{The long-lived modes of interest are high-energy tails of the distribution. Since $f_p {\ll} 1$ when $p^0 {\rightarrow} \infty$, quantum statistics may be neglected.} particles of mass $m$ (which may vanish), with kinetic distribution $f_p(x^\mu)$. The relativistic Boltzmann equation \cite{Groot1980RelativisticKT} has the form $p^\mu \partial_\mu f_p=\mathcal{C}[f_p]$, where $\mathcal{C}$ is some collision functional (it may be Boltzmann's collision integral or some approximation thereof \cite{Cercignani}). Fixed a uniform equilibrium state $f_p^{\text{eq}}=e^{\alpha +\beta_\nu p^\nu}$ ($\alpha$, $\beta_\nu$ constant), we make the decomposition $f_p=f_p^{\text{eq}}(1+\phi_p)$, and we linearise in $\phi_p$. This results in the evolution equation $p^\mu \partial_\mu \phi_p=L\phi_p$, where $L$ is a linear operator (whose domain of definition will be specified later). Fixed one reference frame, not necessarily at rest with respect to the gas, we restrict our attention to the homogenous problem\footnote{The main results extend by continuity to states with $\phi_p \propto e^{ikx^1}$ ($k$ small), because the operator $p^1/p^0$ is bounded \cite[Chapter 5, \S 4.3]{Kato_book}.}, thereby dropping the term $p^j \partial_j \phi_p$. Hence, have
\begin{equation}\label{frittata}
 \partial_t \phi_p =  (p^0)^{-1} L \phi_p \spc \Longrightarrow \spc \phi_p(t)=e^{(p^0)^{-1}L t}\phi_p(0) \, .
\end{equation}
Our goal in this Letter is to discuss the properties of the operator $(p^0)^{-1}L$, which generates the homogenous evolution.
We stress the importance of keeping the factor $(p^0)^{-1}$, since the spectral properties of $L$ may be very different from those of $(p^0)^{-1}L$. For example, if we take $L=-1$, whose spectrum is just $\{-1\}$,  the spectrum of $(p^0)^{-1}L=-(p^0)^{-1}$ is continuous, covering the interval $[-m^{-1},0]$.

\noindent \textit{Features of the collision operator -} We require $L$ to have the following properties (taking $\phi_p \in \mathbb{C}$ for later convenience):
\begin{flalign}
\text{Equilibrium states:}    && \label{Conservation}  & L 1 =L p^\nu=0  \, ,  &\\  
\text{Dissipation:}    && \label{2ndlaw}   
 & \int \dfrac{d^3 p}{(2\pi)^3 p^0} f_p^{\text{eq}} \phi_p^* L \phi_p  \leq 0 \, , &\\ 
\text{Onsager Symmetry:}    && \label{Onsager}   & \bigg[\int \dfrac{d^3 p}{(2\pi)^3 p^0} f_p^{\text{eq}} \psi_p^* L \phi_p\bigg]^*  = \int \dfrac{d^3 p}{(2\pi)^3 p^0} f_p^{\text{eq}} \phi_p^* L \psi_p \, , & 
\end{flalign}
all of which are fulfilled by Boltzmann's collision integral \cite{cercignani_book} (see Supplementary Material for the proof). Equation \eqref{Conservation} expresses the requirement that all equilibrium states $f_p=e^{\tilde{\alpha} +\tilde{\beta}_\nu p^\nu}$ (with $\tilde{\alpha}$ and $\tilde{\beta}_\nu$ constant) must be solutions of the Boltzmann equation. In fact, taking $\tilde{\alpha}=\alpha{+}a$ and $\tilde{\beta}_\nu = \beta_\nu {+}b_\nu$, and linearizing in $a$ and $b_\nu$, we obtain $\phi_p=a{+}b_\nu p^\nu$ which, plugged into the linearised Boltzmann equation, gives $L[a{+}b_\nu p^\nu]=p^\mu \partial_\mu (a{+}b_\nu p^\nu)=0$. To derive \eqref{2ndlaw}, one expresses the entropy production in terms of the information current, $\varsigma=-\partial_\mu E^\mu \geq 0$, and imposes the second law of thermodynamics, $\varsigma \geq 0$. In kinetic theory, the information current $E^\mu$ is known \cite{GavassinoCausality2021,RochaGavassinoFluctu:2024afv}, and we obtain
\begin{equation}
    E^\mu =\dfrac{1}{2} \int \dfrac{d^3 p}{(2\pi)^3 p^0} f_p^{\text{eq}} p^\mu |\phi_p|^2 \, , \spc
     \varsigma = - \mathfrak{Re} \int \dfrac{d^3 p}{(2\pi)^3 p^0} f_p^{\text{eq}} \phi_p^* L \phi_p \, .
\end{equation}
Equation \eqref{Onsager} arises from microscopic time-reversal invariance \cite{GavassinoNonHydro2022,GavassinoCasmir2022}, and allows us to drop the real part in the formula for $\varsigma$. Combining \eqref{Conservation} and \eqref{Onsager}, one finds that the perturbed particle current and stress-energy tensor,
\begin{equation}\label{cool!}
\delta J^\mu = \int \dfrac{d^3 p}{(2\pi)^3p^0} f_p^{\text{eq}} p^\mu \phi_p \, , \spc
\delta T^{\mu \nu} = \int \dfrac{d^3 p}{(2\pi)^3 p^0}  f_p^{\text{eq}} p^\mu p^\nu \phi_p \, ,
\end{equation}
automatically obey the conservation laws $\partial_\mu \delta J^\mu=\partial_\mu \delta T^{\mu \nu}=0$.

\noindent \textit{A convenient Hilbert space -} To establish rigorous results, we must fix a space of functions. We choose the (separable) Hilbert space $\mathcal{H}=L^2(\mathbb{R}^3, f^{\text{eq}}_p)$, namely the space of (complex) square-integrable functions on $\mathbb{R}^3$ with respect to the measure $f^{\text{eq}}_p d^3 p$. The associated inner product is
\begin{equation}\label{sguec}
    (\psi_p,\phi_p) = \int \dfrac{d^3 p}{(2\pi)^3} f_p^{\text{eq}} \psi_p^* \phi_p \, ,
\end{equation}
with corresponding norm $|| \phi_p ||=\sqrt{(\phi_p,\phi_p)}=\sqrt{2E^0}$. The space $\mathcal{H}$ is the set of homogeneous states with finite free energy density \cite{GavassinoGENERIC2022,Mullins:2023tjg,Mullins:2023bur}, and we use the free energy as a square norm\footnote{Textbooks \cite{Groot1980RelativisticKT,cercignani_book} use a different inner product, replacing our measure $f^{\text{eq}}_p d^3p$ with the invariant measure $f^{\text{eq}}_p d^3 p/p^0$. This convention presents a serious problem: If the initial state has finite norm in $L^2(\mathbb{R}^3,f^{\text{eq}}_p/p^0)$, there is no guarantee that such norm will remain finite at later times. By contrast, the second law forces $E^0(t)$ to be non-increasing. Thus, if $\phi_p(0)\in \mathcal{H}$, then $\phi_p(t)\in \mathcal{H}$ for all $t \geq 0$. Some readers may feel that, in relativity, the invariant measure $d^3 p/p^0$ is ``more natural''. This feeling is an artifact of the homogenous limit. Thinking of the inhomogeneous case, it is clear that the norm depends on the Cauchy surface $\Sigma$ upon which we define the state, and thus should \textit{not} be Lorentz-invariant \cite{Hishcock1983,GerochLindblom1990,Geroch_Lindblom_1991_causal}. Indeed, using $2E^0$ as a squared norm is more natural, since it generalizes to $2\int_\Sigma E^\mu d\Sigma_\mu$ in the inhomogeneous case, and we can use Gauss' theorem \cite{MTW_book} to link norms on different Cauchy surfaces: $|| \phi_p ||^2_{\Sigma_2}-|| \phi_p ||^2_{\Sigma_1}=-\int 2\varsigma d^4 x$.} \cite{Dudynski1988}. Clearly, any physically meaningful perturbation $\phi_p$ should belong to $\mathcal{H}$. Indeed, all the products $p^{\nu_1}p^{\nu_2}...p^{\nu_l}$ belong to $\mathcal{H}$, meaning that all the moments $\delta \rho^{0\nu_1 \nu_2...\nu_l}=(p^{\nu_1}p^{\nu_2}...p^{\nu_l},\phi_p)$ are well-defined inside $\mathcal{H}$, including $\delta J^0=(1,\phi_p)$ and $\delta T^{0\nu}=(p^\nu,\phi_p)$.

The inner product \eqref{sguec} allows us to rewrite  \eqref{2ndlaw} and \eqref{Onsager} as follows:
 \begin{equation}\label{bazinga}
    \varsigma= -\bigg(\phi_p , \dfrac{1}{p^0}L \phi_p \bigg)\geq 0 \, , \spc \bigg(\psi_p , \dfrac{1}{p^0}L \phi_p \bigg)^*= \bigg(\phi_p, \dfrac{1}{p^0}L\psi_p \bigg) \, ,
\end{equation}
meaning that the operator $(p^0)^{-1} L$ is Hermitian negative-semidefinite. Furthermore, for $(p^0)^{-1} L$ to be physically meaningful, it must be densely defined in $\mathcal{H}$, since all functions in $\mathcal{H}$ are acceptable physical states\footnote{For sufficiently regular cross-sections, Boltzmann's collision operator is well defined on $C^\infty$ functions with compact support. These are dense in $\mathcal{H}$, making Boltzmann's operator densely defined, see e.g. \cite{Dudynski1988,Strain2010,Speck2011} and references therein. In the following, we will adopt the usual physics convention of saying that  ``statement $X$ holds in Hilbert space $Y$'' whenever $X$ holds in a \textit{dense linear subset} of $Y$.}. The above facts allow us to apply the Friederichs extension (Theorem 5.1.13 of \cite{Pedersenbook}), and promote $(p^0)^{-1}L$ to a self-adjoint operator.

\noindent \textit{A key theorem -} Let us decompose $\mathcal{H}=\mathcal{H}_{\text{hy}} \bigoplus \mathcal{H}_{\text{non-hy}}$, where $\mathcal{H}_{\text{hy}}$ is the space of states that are described by hydrodynamics, and $\mathcal{H}_{\text{non-hy}}$ is the space of non-hydrodynamic excitations. Since we are working in the homogeneous limit, all hydrodynamic ``waves'' take the form of global equilibria $\phi_p=a+b_\nu p^\nu$, so that $\mathcal{H}_{\text{hy}}=\text{span}\{1,p^\nu\}$.
Conversely, the non-hydrodynamic excitations are those states that are ``invisible'' to hydrodynamics. Since hydrodynamics is only aware of the conserved densities, the non-hydrodynamic excitations are the states such that $\delta J^0=\delta T^{0\nu}=0$. Recalling that $\delta J^0=(1,\phi_p)$ and $\delta T^{0\nu}=(p^\nu,\phi_p)$, we conclude that $\mathcal{H}_{\text{non-hy}}=\mathcal{H}_{\text{hy}}^\perp$.
Therefore, $\mathcal{H}_{\text{non-hy}}$ is a Hilbert space in its own right. Furthermore, $\mathcal{H}_{\text{non-hy}}$ is an invariant space of $(p^0)^{-1}L$. This follows from the second equation of \eqref{bazinga}: Just take $\phi_p \in \mathcal{H}_{\text{non-hy}}$ and $\psi_p=a+b_\nu p^\nu$, so that, using \eqref{Conservation},
\begin{equation}
  \bigg(a{+}b_\nu p^\nu , \dfrac{1}{p^0}L \phi_p \bigg)^*= \bigg(\phi_p, \dfrac{1}{p^0}L[a{+}b_\nu p^\nu]  \bigg) =0 \spc \Longrightarrow  \spc  \frac{1}{p^0} L \phi_p \in \mathcal{H}_{\text{hy}}^\perp=\mathcal{H}_{\text{non-hy}} \, .
\end{equation}
Thus, we can restrict the operator $(p^0)^{-1} L$ to a self-adjoint operator on the Hilbert subspace $\mathcal{H}_{\text{non-hy}}$ of pure deviations from local thermodynamic equilibrium. Then, we have the following result, whose proof is provided in the Supplementary Material.

\begin{theorem}[\textbf{Gaplessness Criteria}]\label{theo1}
Suppose that $1$ and $p^\nu$ are the only collisional invariants belonging to $\mathcal{H}$, i.e. $\textup{Ker}(L)= \mathcal{H}_{\textup{hy}}$.  Then, the following facts are equivalent:
\begin{itemize}
\item[\textup{(a)}] We can construct non-hydrodynamic excitations that survive as long as we want, i.e. for any time $t>0$,
\begin{equation}\label{primocriterio}
   \sup_{\phi_p \in \mathcal{H}_{\textup{non-hy}}}  \dfrac{|| \phi_p(t)||}{|| \phi_p(0) ||} =1 \, .
\end{equation}
\item[\textup{(b)}] We can construct non-hydrodynamic states $\phi_p$ such that $||(p^0)^{-1}L\phi_p ||$ is arbitrarily smaller than $|| \phi_p ||$, i.e.
\begin{equation}\label{superuzzuz}
    \inf_{\phi_p \in \mathcal{H}_{\textup{non-hy}}} \dfrac{|| (p^0)^{-1} L \phi_p ||}{||\phi_p||} =0 \, .
\end{equation} 
\item[\textup{(c)}] We can construct non-hydrodynamic states that produce negligibly small entropy, i.e.
\begin{equation}\label{entropona}
    \inf_{\phi_p \in \mathcal{H}_{\textup{non-hy}}} \dfrac{\varsigma}{E^0} =0 \, .
\end{equation}
\item[\textup{(d)}] Working within the space $\mathcal{H}_{\textup{non-hy}}$, the frequency $\omega=0$ is an accumulation point of the spectrum of $(p^0)^{-1}L$. Namely, for any $\epsilon>0$, there is a non-hydrodynamic frequency $\omega \neq 0$ such that $|\omega|<\epsilon$.
\end{itemize}
\end{theorem}

Statements (a,b,c,d) are four alternative (necessary and sufficient) criteria for $\tau_{\text{non-hy}}$ to be infinite. If any one of them is met, all four are met. Criterion (d) is the ``cut'' option considered in \cite{Moore:2018mma}. The implication (d)$\rightarrow$(a) is a rigorous formulation of the ``dehydrodynamization'' mechanism discussed in \cite{Kurkela:2017xis}. Note that, since \eqref{frittata} is linear, the states in \eqref{primocriterio} may be rescaled by an arbitrary constant factor. Thus, if (a) holds, there are long-lived non-hydrodynamic perturbations whose amplitude is (and will forever remain) much larger than an arbitrary preassigned soundwave.

The real novelty of Theorem \ref{theo1} are criteria (b,c). In fact, the recent literature tries to assess (d) by direct evaluation of the spectrum of $(p^0)^{-1}L$ \cite{Moore:2018mma,RochaCut2024cge}. This task requires numerical techniques that are by nature approximate, and sensitive to the details of $L$. Criteria (b) and (c) remarkably simplify the problem, since now one only needs to engineer a sequence $\{\phi_p^{(n)} \}_{n\in \mathbb{N}} \subset \mathcal{H}$ with unit norm, vanishing conserved densities, $(1,\phi_p^{(n)}){=}(p^\nu,\phi_p^{(n)}){=}0$, and such that either $ || (p^0)^{-1}L \phi_p^{(n)} ||$ or $\varsigma[\phi_p^{(n)}]$ converges to zero as $n{\rightarrow} \infty$. Below, we provide an explicit example.

\noindent \textit{Application: Massless scalar particles -} Let us apply Theorem \ref{theo1} to a system of massless classical particles with $\lambda \varphi^4$ interaction (working in the equilibrium rest frame, so $\beta^\mu = \beta \delta^\mu_0$). In this model, the spectrum of Boltzmann's collision operator $L$ is known analytically, and $\textup{Ker}(L)= \mathcal{H}_{\textup{hy}}$ \cite{Denicol:2022bsq}. We consider the sequence
\begin{equation}\label{phipnphi4}
    \phi_p^{(n)} =L_n^{(5)}(\beta p^0) p^{2} p^{3} \,  ,
\end{equation}
where $L^{(2\ell+1)}_n$ are Laguerre polynomials. These states are eigenvectors of $L$, with eigenvalues
\begin{equation}
\chi_{n2} \propto -  \dfrac{n{+}1}{n{+}3} \, . 
\end{equation}
Showing that $(1,\phi_p^{(n)})=(p^\nu,\phi_p^{(n)})=0$ is straightforward. Furthermore, in the Supplementary Material we verify that
\begin{equation}
    \dfrac{|| (p^0)^{-1} L\phi_p^{(n)} ||}{|| \phi_p^{(n)} ||} \leq  \dfrac{\text{const}}{\sqrt{n{+}3}}  \xrightarrow[]{n\rightarrow +\infty} 0.
\end{equation}
Thus, criterion (b) is met. Applying Theorem \ref{theo1}, we conclude from (a) that the gas possesses non-hydrodynamic excitations that survive longer than soundwaves. We also conclude from (d) that the non-hydrodynamic sector is gapless, in agreement with equation \eqref{poWeer}, and with numerical studies \cite{Moore:2018mma,RochaCut2024cge}.

\noindent \textit{High-Energy Criterion -} For Boltzmann's equation, $\textup{Ker}(L)$ always equals $\mathcal{H}_{\textup{hy}}$, and the entropy production is \cite{Groot1980RelativisticKT}
\begin{equation}\label{varsigma}
    \varsigma = \dfrac{1}{8} 
\int \dfrac{d^3p}{(2\pi)^3 p^0} \dfrac{d^3p'}{(2\pi)^3p'^0} \dfrac{d^3q}{(2\pi)^3q^0} \dfrac{d^3q'}{(2\pi)^3q'^0} \, s \upvarsigma(s,\theta) \, \delta^4 (p{+}p'{-}q{-}q') \, f^{\text{eq}}_{p} f^{\text{eq}}_{p'} |\phi_p+\phi_{p'}-\phi_q-\phi_{q'}|^2  \, ,
\end{equation}
where $\upvarsigma(s,\theta)$ is the differential cross-section, $s =E^2_{\text{CM}}$ is the Mandelstam variable, and $\theta$ is the scattering angle in the center of momentum frame. Now, suppose that there are two constants $A\geq 0$ and $0<a \leq 1$ such that $\upvarsigma(s,\theta) \leq A/s^a$. Then, if we consider the sequence of perturbations
\begin{equation}\label{sequenceiaziomain}
    \phi_p^{(n)}= e^{\frac{\beta}{2}(1-n^{-1})p^0} \dfrac{2p^2 p^3}{(p^2)^2{+}(p^3)^2} \, ,
\end{equation}
which are orthogonal to $1$ and $p^\nu$, we find that (see Supplementary Material)
\begin{equation}
    \dfrac{\varsigma[\phi_p^{(n)}]}{E^0[\phi_p^{(n)}]} \leq \dfrac{\text{const}}{n^a} \xrightarrow[]{n\rightarrow +\infty} 0 \, .
\end{equation}
Hence, condition (c) of Theorem \ref{theo1} is fulfilled. Considering that $\varsigma$ is linear in $\upvarsigma$, we 
have the following rigorous result.
\begin{theorem}[\textbf{High-energy criterion}]\label{theo2}
Consider a classical gas of relativistic particles, governed by Boltzmann's equation, with non-vanishing cross-section. Suppose that there are three constants $A\geq 0$, $B \geq 0$,  and $0<a< 1$ such that
\begin{equation}\label{surprising}
\upvarsigma(s,\theta) \leq \dfrac{A}{s^a} + \dfrac{B}{s} \, ,
\end{equation}
for all $\theta$ and $s$. Then, the assumptions of Theorem \ref{theo1} are met, and facts \textup{(a,b,c,d)} occur.
\end{theorem}
This means that, if $\upvarsigma$ decays for $s\rightarrow +\infty$ like a power law $s^{-a}$ (or faster), we can construct non-hydrodynamic excitations that survive as long as we want. The second term in \eqref{surprising} is there to remind us that $\upvarsigma(s,\theta)$ is allowed to diverge at $s= 0$ (which occurs only for massless particles), but it should not grow faster than $s^{-1}$ (for our proof to work). Note that Theorem \ref{theo2} only provides sufficient conditions for (a,b,c,d) to hold. In the Supplementary Material, we generalize Theorem \ref{theo2} to all soft interactions, following the definition of \cite{Dudynski1988}.

\noindent \textit{Application: Screened gauge theories -} Most scattering differential cross-sections in QED \cite{Minten:1969} and QCD \cite{Arnold2003} decay like $1/s$ at high energies, see e.g. the (ultrarelativistic) electron-electron and the gluon-gluon scatterings:
\begin{equation}\label{gaugiuzzo}
\upvarsigma_{ee \rightarrow ee} \propto  \dfrac{(3{+}\cos^2 \! \theta)^2}{s \, \sin^4 \! \theta} \, , \spc
\upvarsigma_{gg \rightarrow gg} \propto  \dfrac{1}{s} \bigg[ 3-\dfrac{\sin^2 \! \theta}{4} +  \dfrac{4{+}12\cos^2 \! \theta}{\sin^4 
 \! \theta} \bigg] \, . 
\end{equation}
The factor $(\sin \theta)^{-4}$ makes the operator $L$ small-angle divergent, reflecting the long-range nature of the interaction, which makes Boltzmann's assumption of local collisions invalid. However, if we correct the cross-sections \eqref{gaugiuzzo} with medium-dependent effects (e.g., Debye screening \cite{bellan_2006}) the small-angle behavior is regularized \cite{messiah1961quantum,Arnold2003}. With certain regularizations, e.g. $(\sin  \theta)^{-4} \rightarrow (\epsilon^2{+}\sin^2 \! \theta)^{-2}$ (with $\epsilon=\text{const}$), Theorem \ref{theo2} applies.

\newpage

\noindent \textit{Application: Yukawa theory -} Ultrarelativistic fermions $f$ interacting through either Yukawa coupling $ \varphi \Bar{f}f$ or $\varphi \Bar{f}\gamma^5 f$ have scattering cross-section $\upvarsigma_{ff \rightarrow ff} \sim 1/s$ \cite{Srednicki:2007qs,Al-Bakri:2023kcr} (the dependence on the angle disappears at high energies).
Theorem \ref{theo2} applies.

\noindent \textit{Final Remarks -} Interactions that decay to zero at high energy are the most common in relativistic physics. Nearly all ``textbook'' (two-body) cross-sections decay like $s^{-1}$, as one would expect from naive dimensional analysis \cite{Peskin_book}. Indeed, $\upvarsigma \propto s^{-1}$ is the universal scaling law of those theories that ``forget'' their reference mass scales in the ultrarelativistic limit. In those theories that ``remember'' some mass scale $M$ at high energy, unitarity requires that $\sigma \leq \pi M^{-2}\ln^2(s)$ \cite{Castaldi:1983hrk}, which allows for mild growth of $\upvarsigma$ with $s$ \cite{Kupsch:1982aa}. Hence, some diluted gases with a non-hydrodynamic gap are likely to exist. However, they seem to be the exception, not the rule.

Let us provide two arguments suggesting that the non-hydrodynamic sector of QCD plasmas is probably gapless.
\begin{itemize}
\item All relaxation-type models of QCD matter adopt an energy-dependent relaxation time $\tau \propto (p^0)^a$, with $0<a< 1$ \cite{Dusling:2009df,Kurkela:2017xis,Rocha:2021zcw}. If this approximation is valid, at least qualitatively, we can use it to learn about the scaling law of the cross sections. In fact, for a high-energy particle interacting with a thermal bath, $s \propto p^0$, and $\tau \propto \sigma^{-1}$. We conclude that $\sigma \propto 1/s^a$, consistently with \eqref{surprising}. Hence, in a corresponding ``full-fledged'' Boltzmann equation, Theorem \ref{theo2} applies, and the non-hydrodynamic sector is gapless. 
\item According to \cite{Baier:1996kr,Baier:1996sk}, particles with $p^0 \gg T$ traveling through a QCD plasma lose energy following the equation $\dot{p}^0=-2C\sqrt{p^0}$, whose solution is $p^0(t)=\big[\sqrt{p^0(0)}{-}Ct \big]^2\Theta[\sqrt{p^0(0)}{-}Ct \big]$. Suppose that there is a diluted population of high-energy particles whose number (per unit energy) falls like $(p^0)^{-3}$ at $t=0$. Then, at late times, the non-thermal energy carried by these particles decays slower than exponentially,
\begin{equation}
E_{\text{non-th}}(t) \propto \int_{(Ct)^2}^{\infty} \dfrac{(\sqrt{p^0}-Ct)^2}{(p^0)^3} dp^0 = \dfrac{1}{6(Ct)^2} \, ,
\end{equation}
proving the absence of a gap.
\end{itemize}

Although our analysis was formally carried out within the linear theory, the discussed phenomena are not artifacts of linearization. In fact, in the Supplementary Material, we show that, for each linear state considered here, there is a family of distribution functions that behave under the non-linear theory in a way that is indistinguishable from the linear theory prediction.

\section*{Acknowledgements}

This work was supported by a Vanderbilt's Seeding Success Grant. I would like to thank G. Rocha, G. Denicol, J. Noronha, M. Disconzi, J. Speck, J.F. Paquet, I. Danhoni, and K. Ingles for useful discussions.

\bibliography{Biblio}

\onecolumngrid
\newpage
\begin{center}
  \textbf{\large Gapless non-hydrodynamic modes in relativistic kinetic theory\\ Supplementary Material}\\[.2cm]
  L. Gavassino\\[.1cm]
  {\itshape Department of Mathematics, Vanderbilt University, Nashville, TN, USA\\}
(Dated: \today)\\[1cm]
\end{center}

\setcounter{equation}{0}
\setcounter{figure}{0}
\setcounter{table}{0}
\setcounter{page}{1}
\renewcommand{\theequation}{S\arabic{equation}}
\renewcommand{\thefigure}{S\arabic{figure}}

\title{Gapless non-hydrodynamic modes in relativistic kinetic theory: Supplementary Material}
\author{L.~Gavassino}
\email{lorenzo.gavassino@vanderbilt.edu}
\affiliation{Department of Mathematics, Vanderbilt University, Nashville, TN 37211, USA}

\vspace{-1cm}
\section{Long-lived exponential tails}
\vspace{-0.2cm}

\subsection{Transient evolution of a $\lambda \varphi^4$ gas perturbed by a shear deformation}
\vspace{-0.2cm}

Consider a classical gas of massless particles. 
Suppose that such gas has been prepared in the following state:
\begin{equation}\label{exciuto}
f_p= e^{\alpha -\beta p^0} \bigg[ 1+\epsilon \beta \dfrac{p^2 p^3}{p^0}\bigg] +\mathcal{O}(\epsilon^2) \, ,
\end{equation}
where $\epsilon$ is a small parameter. This state describes the transient linear response of a gas that was externally perturbed by a shear rate $\sigma_{23} \, {\propto}  \, \delta (t)$ \cite[Eq.s (48) and (49)]{Denicol_Relaxation_2011}. The term $\mathcal{O}(\epsilon^2)$ is the non-linear part of the response, which is beyond the scope of this article.
The particle current and stress-energy tensor in this state are
\begin{equation}
J^\mu = \int \dfrac{d^3 p}{(2\pi)^3 p^0} p^\mu f_p = \dfrac{e^\alpha T^3}{\pi^2} 
\begin{bmatrix}
1 \\
0 \\
0 \\
0 \\
\end{bmatrix} +\mathcal{O}(\epsilon^2) \, , \quad \quad 
T^{\mu \nu} = \int \dfrac{d^3 p}{(2\pi)^3 p^0} p^\mu  p^\nu f_p = \dfrac{e^\alpha T^4}{\pi^2} 
\begin{bmatrix}
3 & 0 & 0 & 0 \\
0 & 1 & 0 & 0 \\
0 & 0 & 1 & 4\epsilon/5 \\
0 & 0 & 4\epsilon/5 & 1 \\
\end{bmatrix} +\mathcal{O}(\epsilon^2) \, .
\end{equation}
We see that, to first order, the variable $\epsilon$ appears only in the expression for $T^{23}$. Hence, the linear perturbation $e^{\alpha -\beta p^0}  \beta p^2 p^3/p^0$ corresponds, from a fluid-dynamic perspective, to an initial anisotropy of the stress tensor. Such anisotropy is a pure non-hydrodynamic mode, since it carries no net conserved density (i.e. $J^0$ and $T^{0 \nu}$ do not depend on $\epsilon$). For this reason, when we evolve this state using the linear theory, the function $T^{23}(t)$ will decay to zero, and the gas will approach equilibrium. 

Let us determine the decay law of $T^{23}$. To this end, we first evolve the distribution function $f_p$ in the relaxation time approximation, with an energy-dependent relaxation time. Again, we ``mimic'' the evolution of a self-interacting scalar field theory by taking $\tau_p=\beta p^0 \tau$ (where $\tau$ is a constant) \cite{Denicol:2022bsq}. The state at time $t\geq 0$ is, therefore,
\begin{equation}\label{exciutot}
f_p(t)= e^{\alpha -\beta p^0} \bigg[ 1+\epsilon \beta \dfrac{p^2 p^3}{p^0} e^{-t/(\beta p^0 \tau)}\bigg] +\mathcal{O}(\epsilon^2) \, ,
\end{equation}
The integral defining the function $T^{23}(t)$ can be solved analytically. The result is (in the limit of small $\epsilon$)
\begin{equation}
 \lim_{\epsilon \rightarrow 0} \, \,  \dfrac{T^{23}(t)}{T^{23}(0)}= \dfrac{(t/\tau)^{5/2}}{12} K_5 \big(2\sqrt{t/\tau}\big) \, ,
\end{equation}
where $K_5$ is the modified Bessel function of the second kind. At late time, the anisotropy approaches the following asymptotic behavior:
\begin{equation}\label{sblurfono}
\lim_{\epsilon \rightarrow 0}\,  \, \dfrac{T^{23}(t)}{T^{23}(0)}\approx \dfrac{\sqrt{\pi}}{24} (t/\tau)^{9/4} e^{-2\sqrt{t/\tau}} \, .    
\end{equation}
As can be seen, the decay law of the shear stress is subexponential, meaning that the non-hydrodynamic excitation $e^{\alpha -\beta p^0}  \beta p^2 p^3/p^0$ lives longer than any hydrodynamic wave (see also \cite{Kurkela:2017xis}). 

The above example illustrates two important facts:
\begin{itemize}
\item[1.] We do not need to inject a lot of energy into the system to generate a long-lived mode. We only need to construct excitations with unbounded support in the variable $p$. This happens naturally when the perturbing source is a hydrodynamic gradient, as in the example above. Note that subexponential decay laws like \eqref{sblurfono} have been observed also in stochastic correlation functions with relaxation time $\tau_p=\beta p^0 \tau$ \cite{RochaGavassinoFluctu:2024afv}.
\item[2.] Since the state \eqref{exciuto} describes the response of the gas to shear deformation, the decay law \eqref{sblurfono} should be interpreted as the transient evolution of the shear stress component $\pi^{23}$ induced by an instantaneous change in shear rate $\sigma_{23}$. Clearly, this does not coincide with the Israel-Stewart transient behavior (where $\pi^{23}{\sim} e^{-t/\tau_\pi}$ \cite[\S 6.5.1]{rezzolla_book}). Therefore, ``second-order hydrodynamics'' does not accurately depict how the shear stress relaxes to its Navier-Stokes value. This forces one to revise the physical interpretation of $\tau_\pi$ as the slowest relaxation time \cite{WagnerGavassino2023jgq}. A possible way around this problem might be to introduce a cutoff in $p^0$, as was done in \cite{Gavassino2024InfiniteOrder} (and suggested in \cite{RochaCut2024cge}).
\end{itemize}

\subsection{A bounded choice of $\phi$}

In the previous example, we considered a physically-motivated transient perturbation of the shear stress tensor, arising from an instantaneous shear deformation. This had the inconvenience of being applicable only for $\epsilon$ infinitesimal. In fact, at finite $\epsilon$, the magnitude of $\epsilon p^2 p^3/p^0$ becomes arbitrarily large at high energy, and may eventually go below $-1$. This means that the distribution function can become negative in some regions of phase space, if we do not include a balancing term of order $\epsilon^2$. Let us now give an example of a state where this does not happen, and where the non-equilibrium perturbation is much smaller than the equilibrium part for all momenta, also at finite $\epsilon$.
Specifically, we choose
\begin{equation}\label{exciutone}
f_p= e^{\alpha -\beta p^0} \bigg[ 1+\epsilon  \dfrac{p^2 p^3}{(p^0)^2}\bigg] \, ,
\end{equation}
which is very close to equilibrium for all $(p^1,p^2,p^3)$, provided that $\epsilon \ll 1$.
This disturbance is, for our purposes, completely analogous to \eqref{exciuto}, since it only perturbs the component $T^{23}$ of the stress-energy tensor. In the relaxation-type approximation, with $\tau_p=\beta p^0 \tau$, the distribution function evolves as follows:
\begin{equation}\label{exciutot}
f_p(t)= e^{\alpha -\beta p^0} \bigg[ 1+\epsilon  \dfrac{p^2 p^3}{(p^0)^2} e^{-t/(\beta p^0 \tau)}\bigg] \, .
\end{equation}
The resulting shear stress exhibits the following (exact) time dependence:
\begin{equation}
\dfrac{T^{23}(t)}{T^{23}(0)} =\dfrac{(t/\tau)^2}{3} K_4 (2\sqrt{t/\tau}) \, .
\end{equation}
This, again, decays slower than exponentially at late times, since we have that, asymptotically,
\begin{equation}
\dfrac{T^{23}(t)}{T^{23}(0)} \approx \dfrac{\sqrt{\pi}}{6} (t/\tau)^{7/4} e^{-2\sqrt{t/\tau}} \, .
\end{equation}\\

\newpage

\section{Properties of the linearised collision operator}\label{AppHwrm}

Here, we verify that Boltzmann's collision integral fulfills all the general properties mentioned in the main text. 

The action of $L$ on a generic function $\phi_p$ is
\begin{equation}\label{gargarensis}
    L\phi_p = \int \dfrac{d^3p'}{(2\pi)^3p'^0} \dfrac{d^3q}{(2\pi)^3q^0} \dfrac{d^3q'}{(2\pi)^3q'^0} W_{pp'\leftrightarrow qq'} f^{\text{eq}}_{p'}(\phi_q{+}\phi_{q'}{-}\phi_p{-}\phi_{p'}) \, ,
\end{equation}
where $W_{pp'\leftrightarrow qq'}=W_{p'p\leftrightarrow qq'}=W_{qq'\leftrightarrow pp'} \geq 0$ is the transition rate, which vanishes unless $p{+}p'=q{+}q'$, due to energy-momentum conservation. This implies that $W_{pp'\leftrightarrow qq'} \neq 0$ only if $f^{\text{eq}}_{p} f^{\text{eq}}_{p'}=f^{\text{eq}}_{q} f^{\text{eq}}_{q'}$. 

The fact that $L1=Lp^\nu=0$ is evident. Thus, we only need to prove the properties ``Dissipation'' and ``Onsager Symmetry''. To do that, we multiply \eqref{gargarensis} by $f^{\text{eq}}_p \psi^*_p/p^0$, and integrate over all momenta:
\begin{equation}\label{origEz}
 \int \dfrac{d^3p}{(2\pi)^3 p^0} f^{\text{eq}}_{p} \psi_p^* L\phi_p=-
\int \dfrac{d^3p}{(2\pi)^3 p^0} \dfrac{d^3p'}{(2\pi)^3p'^0} \dfrac{d^3q}{(2\pi)^3q^0} \dfrac{d^3q'}{(2\pi)^3q'^0} W_{pp'\leftrightarrow qq'} f^{\text{eq}}_{p} f^{\text{eq}}_{p'} \psi_p^*(\phi_p{+}\phi_{p'}{-}\phi_q{-}\phi_{q'})
\end{equation}
Since the integration variables $(pp'qq')$ are dummy variables, we can rename them at will. If we perform the change of variables $(pp'qq') \rightarrow (p'pqq')$ and use the symmetry $W_{p'p\leftrightarrow qq'}=W_{pp'\leftrightarrow qq'}$, we obtain
\begin{equation}\label{origEz2}
     \int \dfrac{d^3p}{(2\pi)^3 p^0} f^{\text{eq}}_{p} \psi_p^* L\phi_p=-
\int \dfrac{d^3p}{(2\pi)^3 p^0} \dfrac{d^3p'}{(2\pi)^3p'^0} \dfrac{d^3q}{(2\pi)^3q^0} \dfrac{d^3q'}{(2\pi)^3q'^0} W_{pp'\leftrightarrow qq'} f^{\text{eq}}_{p} f^{\text{eq}}_{p'} \psi_{p'}^*(\phi_p{+}\phi_{p'}{-}\phi_q{-}\phi_{q'}) \, .
\end{equation}
If, instead, we rename the variables in \eqref{origEz} as follows: $(pp'qq')\rightarrow (qq'pp')$, we invoke the symmetry $W_{pp'\leftrightarrow qq'}=W_{qq'\leftrightarrow pp'}$ and the condition that $f^{\text{eq}}_{p} f^{\text{eq}}_{p'}= f^{\text{eq}}_{q} f^{\text{eq}}_{q'}$ whenever the transition is allowed, we obtain
\begin{equation}\label{origEz3}
 \int \dfrac{d^3p}{(2\pi)^3 p^0} f^{\text{eq}}_{p} \psi_p^* L\phi_p=-
\int \dfrac{d^3p}{(2\pi)^3 p^0} \dfrac{d^3p'}{(2\pi)^3p'^0} \dfrac{d^3q}{(2\pi)^3q^0} \dfrac{d^3q'}{(2\pi)^3q'^0} W_{pp'\leftrightarrow qq'} f^{\text{eq}}_{p} f^{\text{eq}}_{p'} (-\psi_q^*)(\phi_p{+}\phi_{p'}{-}\phi_q{-}\phi_{q'})\, .
\end{equation}
Finally, let's perform the change of variables $(pp'qq') \rightarrow (pp'q'q)$ in \eqref{origEz3}, and use the symmetry $W_{pp'\leftrightarrow q'q}=W_{pp'\leftrightarrow qq'}$. The result is 
\begin{equation}\label{origEz4}
 \int \dfrac{d^3p}{(2\pi)^3 p^0} f^{\text{eq}}_{p} \psi_p^* L\phi_p=-
\int \dfrac{d^3p}{(2\pi)^3 p^0} \dfrac{d^3p'}{(2\pi)^3p'^0} \dfrac{d^3q}{(2\pi)^3q^0} \dfrac{d^3q'}{(2\pi)^3q'^0} W_{pp'\leftrightarrow qq'} f^{\text{eq}}_{p} f^{\text{eq}}_{p'} (-\psi_{q'}^*)(\phi_p{+}\phi_{p'}{-}\phi_q{-}\phi_{q'})\, .
\end{equation}
Adding up \eqref{origEz}, \eqref{origEz2}, \eqref{origEz3}, and \eqref{origEz4}, we finally obtain
\begin{equation}
  \int \dfrac{d^3p}{(2\pi)^3 p^0} f^{\text{eq}}_{p} \psi_p^* L\phi_p=-\dfrac{1}{4} 
\int \dfrac{d^3p}{(2\pi)^3 p^0} \dfrac{d^3p'}{(2\pi)^3p'^0} \dfrac{d^3q}{(2\pi)^3q^0} \dfrac{d^3q'}{(2\pi)^3q'^0} W_{pp'\leftrightarrow qq'} f^{\text{eq}}_{p} f^{\text{eq}}_{p'} (\psi^*_p{+}\psi^*_{p'}{-}\psi^*_q{-}\psi^*_{q'})(\phi_p{+}\phi_{p'}{-}\phi_q{-}\phi_{q'}) \, .
\end{equation}
The properties of $L$ are now manifest. In fact, setting $\psi=\phi$, we obtain
\begin{equation}
  \int \dfrac{d^3p}{(2\pi)^3 p^0} f^{\text{eq}}_{p} \psi_p^* L\phi_p=-\dfrac{1}{4} 
\int \dfrac{d^3p}{(2\pi)^3 p^0} \dfrac{d^3p'}{(2\pi)^3p'^0} \dfrac{d^3q}{(2\pi)^3q^0} \dfrac{d^3q'}{(2\pi)^3q'^0} W_{pp'\leftrightarrow qq'} f^{\text{eq}}_{p} f^{\text{eq}}_{p'} |\phi_p{+}\phi_{p'}{-}\phi_q{-}\phi_{q'}|^2 \leq 0 \, ,
\end{equation}
thereby proving ``Dissipation''. To verify ``Onsager Symmetry'', we can write its two sides explicitly:
\begin{equation}
\begin{split}
 \bigg[ \int \dfrac{d^3p}{(2\pi)^3 p^0} f^{\text{eq}}_{p} \psi_p^* L\phi_p \bigg]^*={}& {-}\dfrac{1}{4} 
\int \dfrac{d^3p}{(2\pi)^3 p^0} \dfrac{d^3p'}{(2\pi)^3p'^0} \dfrac{d^3q}{(2\pi)^3q^0} \dfrac{d^3q'}{(2\pi)^3q'^0} W_{pp'\leftrightarrow qq'} f^{\text{eq}}_{p} f^{\text{eq}}_{p'} (\psi_p{+}\psi_{p'}{-}\psi_q{-}\psi_{q'})(\phi^*_p{+}\phi^*_{p'}{-}\phi^*_q{-}\phi^*_{q'}) \, , \\
  \int \dfrac{d^3p}{(2\pi)^3 p^0} f^{\text{eq}}_{p} \phi_p^* L\psi_p={}& {-}\dfrac{1}{4} 
\int \dfrac{d^3p}{(2\pi)^3 p^0} \dfrac{d^3p'}{(2\pi)^3p'^0} \dfrac{d^3q}{(2\pi)^3q^0} \dfrac{d^3q'}{(2\pi)^3q'^0} W_{pp'\leftrightarrow qq'} f^{\text{eq}}_{p} f^{\text{eq}}_{p'} (\phi^*_p{+}\phi^*_{p'}{-}\phi^*_q{-}\phi^*_{q'})(\psi_p{+}\psi_{p'}{-}\psi_q{-}\psi_{q'}) \, . \\
\end{split}
\end{equation}
Clearly, they coincide.

\newpage

\section{Proof of Theorem 1}

\begin{theorem}\label{theo1}
Suppose that $1$ and $p^\nu$ are the only collisional invariants belonging to $\mathcal{H}$, i.e. $\textup{Ker}(L)= \mathcal{H}_{\textup{hy}}$.  Then, the following facts are equivalent:
\begin{itemize}
\item[\textup{(a)}] We can construct non-hydrodynamic excitations that survive as long as we want, i.e. for any time $t>0$,
\begin{equation}\label{survival}
   \sup_{\phi_p \in \mathcal{H}_{\textup{non-hy}}}  \dfrac{|| \phi_p(t)||}{|| \phi_p(0) ||} =1 \, .
\end{equation}
\item[\textup{(b)}] We can construct non-hydrodynamic states $\phi_p$ such that $||(p^0)^{-1}L\phi_p ||$ is arbitrarily smaller than $|| \phi_p ||$, i.e.
\begin{equation}\label{superuzzuz}
    \inf_{\phi_p \in \mathcal{H}_{\textup{non-hy}}} \dfrac{|| (p^0)^{-1} L \phi_p ||}{||\phi_p||} =0 \, .
\end{equation} 
\item[\textup{(c)}] We can construct non-hydrodynamic states that produce negligibly small entropy, i.e.
\begin{equation}\label{entropona}
    \inf_{\phi_p \in \mathcal{H}_{\textup{non-hy}}} \dfrac{\varsigma}{E^0} =0 \, .
\end{equation}
\item[\textup{(d)}] Working within the space $\mathcal{H}_{\textup{non-hy}}$, the frequency $\omega=0$ is an accumulation point of the spectrum of $(p^0)^{-1}L$. Namely, for any $\epsilon>0$, there is a non-hydrodynamic frequency $\omega \neq 0$ such that $|\omega|<\epsilon$.
\end{itemize}
\end{theorem}
\begin{proof}
Since statements (a,b,c,d) are properties of the restriction of $(p^0)^{-1}L$ to $\mathcal{H}_{\text{non-hy}}$, we will work inside the Hilbert subspace $\mathcal{H}_{\text{non-hy}}$ across the whole proof.  In the following, we will use only two ingredients. First, that $(p^0)^{-1}L$ is a self-adjoint operator on $\mathcal{H}_{\text{non-hy}}$. Second, that, since $\textup{Ker}(L)\cap \mathcal{H}_{\text{non-hy}}=\{ 0 \}$, and $\textup{Ker}(1/p^0)=\{ 0\}$, the number $0$ is not a proper eigenvalue of $(p^0)^{-1}L$ inside $\mathcal{H}_{\text{non-hy}}$.
This said, let us first prove the chain (b)$\rightarrow$(c)$\rightarrow$(d)$\rightarrow$(b). 

\noindent (b)$\rightarrow$(c): We invoke the Cauchy-Schwartz inequality:
\begin{equation}
 \dfrac{\varsigma}{2E^0} = \dfrac{\bigg| \bigg(\phi_p , \dfrac{1}{p^0}L \phi_p \bigg) \bigg|}{(\phi_p,\phi_p)} \leq  \dfrac{|| \phi_p || \, \bigg| \bigg|  \dfrac{1}{p^0}L \phi_p  \bigg| \bigg|}{|| \phi_p ||^2} = \dfrac{ \bigg| \bigg|  \dfrac{1}{p^0}L \phi_p  \bigg| \bigg|}{|| \phi_p ||} \, .
\end{equation}
Therefore, if there is a sequence of states along which $||(p^0)^{-1}L\phi_p ||/||\phi_p||$ converges to zero, also $\varsigma/E^0$ will tend to zero along such a sequence. Hence, \eqref{superuzzuz} implies \eqref{entropona}.

\noindent (c)$\rightarrow$(d): Equation \eqref{entropona} is equivalent to
\begin{equation}
\sup_{\phi_p \in \mathcal{H}_{\text{non-hy}}}  \dfrac{\bigg(\phi_p , \dfrac{1}{p^0}L \phi_p \bigg) }{(\phi_p,\phi_p)} =0 \, . 
\end{equation}
Since $(p^0)^{-1}L$ is self-adjoint, we invoke Theorem 2.19 of \cite[Chapter 2, \S 2.4]{Teschlbook}, and we conclude that $\sup \text{Sp}[(p^0)^{-1}L]{=}0$, where ``Sp'' denotes the spectrum (within $\mathcal{H}_{\text{non-hy}}$). However, according to Proposition 5.2.13 of \cite[Chapter 5, \S 5.2]{Pedersenbook}, the spectrum of a self-adjoint operator is non-empty and closed. Thus, $0$ belongs to the spectrum, by Theorem 2.28 of \cite[Chapter 2]{RudinAnalysis_book}. However, we know that $0$ is not a proper eigenvalue in $\mathcal{H}_{\text{non-hy}}$, meaning that it does not belong to the point spectrum. Furthermore, $0$ cannot belong to the residual spectrum,  which is empty, by the spectral theorem, see Theorem 1 of \cite[Lecture 18]{Bhatiabook}. Therefore, it must belong to the continuous spectrum, by Point 6 of \cite[Lecture 17]{Bhatiabook}. On the other hand, all isolated points in the spectrum of a self-adjoint operator belong to the point spectrum \cite[Chapter 5, \S 3.5]{Kato_book}. Therefore, since $0$ does not belong to the point spectrum (within $\mathcal{H}_{\text{non-hy}}$), it cannot be isolated, and it must be an accumulation point of the spectrum, proving (d).

\noindent (d)$\rightarrow$(b): Suppose that $0$ is an accumulation point of the spectrum. Then, $0$ belongs to the spectrum, since the spectrum of a self-adjoint operator is closed (again, by Proposition 5.2.13 of \cite[Chapter 5, \S 5.2]{Pedersenbook}). On the other hand, all points in the spectrum of a self-adjoint operator are approximate eigenvalues, by Theorem 1 of \cite[Lecture 18]{Bhatiabook}. Thus, (b) holds, by definition of approximate eigenvalue, see Points 2 and 5 of \cite[Lecture 17]{Bhatiabook}.

Now we only need to connect (a) to one element of (b,c,d).

\noindent (a)$\leftrightarrow$(d): Since $\phi_p(t)=e^{(p^0)^{-1}L t}\phi_p(0)$ (where the exponential is defined through spectral calculus \cite{Yosidabook}), equation \eqref{survival} can be equivalently rewritten as follows:
\begin{equation}
       \sup_{\phi_p \in \mathcal{H}_{\text{non-hy}}}  \dfrac{||e^{(p^0)^{-1}L t}\phi_p||}{|| \phi_p ||} =1 \, .
\end{equation}
This is, in turn, equivalent to saying that $e^{(p^0)^{-1}L t}$ is a bounded operator with norm $||e^{(p^0)^{-1}L t}||=1$. However, since $(p^0)^{-1}L$ is self-adjoint, also $e^{(p^0)^{-1}L t}$ is self-adjoint, by Theorem 3 of \cite[Chapter XI, \S 12]{Yosidabook}. Hence,
\begin{equation}
    || e^{(p^0)^{-1}L t} ||= \sup_{\phi_p \in \mathcal{H}_{\text{non-hy}}}  \dfrac{(\phi_p, e^{(p^0)^{-1}L t}\phi_p)}{(\phi_p,\phi_p)} = \sup_{\mathcal{H}_{\text{non-hy}}}  \text{Sp} \big[e^{(p^0)^{-1}L t}\big] \, ,
\end{equation}
where we have used Lemma 2.14  and Theorem 2.19 of \cite[Chapter 2]{Teschlbook}. Therefore, condition (a) is equivalent to the requirement that the supremum of the spectrum of $e^{(p^0)^{-1}L t}$ be equal to 1. On the other hand, condition (d) is equivalent to the requirement that the supremum of the spectrum of $(p^0)^{-1}L$ be equal to 0. Indeed, the two are related. In fact, applying Lemma 3.12 of \cite[Chapter 3, \S 3.2]{Teschlbook}, and considering that $f(x)=e^{tx}$ is continuous, we have
\begin{equation}
   \text{Sp}[e^{(p^0)^{-1}L t}]= \overline{e^{t \, \text{Sp}[(p^0)^{-1}L]} } \, ,
\end{equation}
which implies (recall that $t>0$)
\begin{equation}\label{infsup}
  \sup_{\mathcal{H}_{\text{non-hy}}}  \text{Sp} \big[e^{(p^0)^{-1}L t}\big] = \exp \bigg\{ t \sup_{ \mathcal{H}_{\text{non-hy}}} \text{Sp}[(p^0)^{-1}L] \bigg\} \, .  
\end{equation}
This shows that (a) [i.e. left-hand side of \eqref{infsup} equals 1] and (d) [i.e. right-hand side of \eqref{infsup} equals $e^0$] are equivalent. This completes our proof.
\end{proof}

\textbf{Remark 1:} Point (d) tells us that the restriction of $(p^0)^{-1}L$ to the non-hydrodynamic space $\mathcal{H}_{\text{non-hy}}$ possesses an infinite list of spectral points that accumulate at $0$. One may ask whether these spectral points ``survive'' when we go back to $\mathcal{H}$. The answer is, indeed, affirmative. In fact, all the spectral points of a self-adjoint operator are approximate eigenvalues. This means that $\lambda$ is a spectral point of the restriction of $(p^0)^{-1}L$ to $\mathcal{H}_{\text{non-hy}}$ if and only if there is a sequence of states $\phi_p^{(n)}\in \mathcal{H}_{\text{non-hy}}$ with norm 1 such that
\begin{equation}\label{appruxux}
    \lim_{n\rightarrow \infty} ||\big[(p^0)^{-1}L-\lambda \big]\phi_p^{(n)} ||=0 \, .
\end{equation}
Clearly, equation \eqref{appruxux} is still fulfilled by the same sequence $\phi_p^{(n)}$ when we regard $(p^0)^{-1}L$ as an operator on $\mathcal{H}$. Thus, $\lambda$ is a spectral point of $(p^0)^{-1}L$ also in $\mathcal{H}$.\\

\textbf{Remark 2:} Theorem 1 was proven, for clarity, under the assumption that the particle number is conserved. However, this assumption may be safely released. One just needs to make sure that conditions (7), (8), and (9) of the main text still hold, with the only exception that $L1$ now is non-vanishing, due to particle creation and annihilation processes. Also, note that, when particles are not conserved, the parameter $\alpha$ in the equilibrium distribution must be set to zero.

\newpage

\section{Massless particles with $\lambda \varphi^4$ interactions}\label{Appubobb}

Here, we prove that an ultrarelativistic gas with $\lambda \varphi^4$ interaction fulfills criterion (b) of Theorem \ref{theo1}.

In massless $\lambda \varphi^4$ kinetic theory (for classical particles), with coupling strength $g$, the states 
\begin{equation}\label{phiF4f}
\phi_p^{(n)} =L^{(5)}_n(\beta p^0) p^2 p^3  
\end{equation}
are eigenvectors of $L$, i.e. $L[L^{(5)}_n(\beta p^0) p^2 p^3]=\chi_{n2} \, L^{(5)}_n(\beta p^0) p^2 p^3$. The corresponding eigenvalues are
\begin{equation}
    \chi_{n2}= -\dfrac{g\mathcal{M}}{2} \dfrac{n+1}{n+3} \, ,
\end{equation}
which are bounded by the inequality $|\chi_{n2}|\leq g\mathcal{M}/2$.

For completeness, let us first prove that $\phi_p^{(n)}\in \mathcal{H}_{\text{non-hy}}$. This is straightforward to show in spherical coordinates, 
\begin{equation}
    p^\nu =
    \begin{bmatrix}
p \\
p \cos \theta \\
p \sin \theta \cos \varphi \\
p \sin \theta \sin \varphi \\
    \end{bmatrix} \, , \spc d^3p=p^2 dp \, \sin \theta d\theta \, d\varphi \, .
\end{equation}
In fact, in these coordinates, $p^{2} p^{3}  \propto \sin(2\varphi)$, and the inner products $(1,\phi_p^{(n)})$, $(p^0,\phi_p^{(n)})$, and $(p^1,\phi_p^{(n)})$ vanish, being proportional to $\int_0^{2\pi}\sin(2\varphi)d\varphi=0$. Also the inner products $(p^2,\phi_p^{n})$ and $(p^3,\phi_p^{n})$ vanish, since
\begin{equation}
    (\phi_p^{(n)}, p^2{+}ip^3) \propto \int_0^{2\pi} e^{i\varphi}\sin(2\varphi) d\varphi =0 \, .
\end{equation}
Therefore, all the functions $\phi_p^{(n)}$ are non-hydrodynamic. Now we need to study the behavior of the norm of $(p^0)^{-1}L\phi_p^{(n)}$ at large $n$. Since $||(p^0)^{-1}L\phi_p^{(n)}||=||\chi_{n2}(p^0)^{-1}\phi_p^{(n)}|| \leq \frac{1}{2}g\mathcal{M}|| (p^0)^{-1}\phi_p^{(n)}||$, we can write
\begin{equation}
\dfrac{|| (p^0)^{-1} L \phi_p^{(n)} ||^2}{||\phi_p^{(n)}||^2}\leq \bigg(\dfrac{g\mathcal{M}}{2}\bigg)^2 \dfrac{|| (p^0)^{-1} L_n^{(5)}(p^0/T) p^{ 2} p^{3 } ||^2}{|| L_n^{(5)}(p^0/T) p^{ 2} p^{3 }||^2} = \bigg(\dfrac{g\mathcal{M}}{2}\bigg)^2 \dfrac{ \displaystyle\int \dfrac{d^3 p}{(2\pi)^3(p^0)^2} e^{-p^0/T} \bigg[ L_n^{(5)}(p^0/T) p^{ 2} p^{3 }\bigg]^2}{\displaystyle\int \dfrac{d^3 p}{(2\pi)^3} e^{-p^0/T} \bigg[L_n^{(5)}(p^0/T) p^{ 2} p^{3 }\bigg]^2} \, .
\end{equation}
Expressing both integrals in spherical coordinates, and simplifying the common factors, we obtain
\begin{equation}
    \dfrac{|| (p^0)^{-1} L \phi_p^{(n)} ||^2}{||\phi_p^{(n)}||^2} \leq \bigg(\dfrac{g\mathcal{M}}{2T}\bigg)^2 \dfrac{ \displaystyle\int_0^{+\infty}   x^4  e^{-x} \big[ L_n^{(5)}(x) \big]^2 dx} {\displaystyle\int_0^{+\infty}   x^6  e^{-x} \big[ L_n^{(5)}(x) \big]^2 dx} =\dfrac{1}{10 (n{+}3)} \bigg(\dfrac{g\mathcal{M}}{2T}\bigg)^2 \xrightarrow[]{n\rightarrow +\infty} 0,
\end{equation}
which is what we wanted to prove.
To evaluate the integrals, we used the following general formula \cite{Mavormatis1990}:
\begin{equation}
    \int_0^{+\infty} x^{\nu} e^{-x} \big[ L_n^{(a)}(x) \big]^2 dx =\binom{n{+}a}{n} \binom{n{+}a{-}\nu{-}1}{n}\Gamma(\nu {+}1) \cdot {_3 F_2}(-n,\nu{+}1,\nu{-}a{+}1;a{+}1,\nu{-}a{-}n{+}1;1) \, ,
\end{equation}
which implies, in our case,
\begin{equation}
\begin{split}
\int_0^{+\infty} x^{4} e^{-x} \big[ L_n^{(5)}(x) \big]^2 dx ={}&  \dfrac{(n+5)!}{n!} \, \dfrac{1}{5} \, ,\\
\int_0^{+\infty} x^{6} e^{-x} \big[ L_n^{(5)}(x) \big]^2 dx ={}& \dfrac{(n+5)!}{n!} \, (2n+6)\, ,\\
\end{split}
\end{equation}
whose ratio is, indeed, $(10n{+}30)^{-1}$.

\newpage

\section{Bound on the entropy production used in the proof of Theorem 2}

We have the following sequence of entropy production rates:
\begin{equation}\label{varsigmona}
    \varsigma_n = \dfrac{(2\pi)^6}{8} 
\int \dfrac{d^3p}{(2\pi)^3 p^0} \dfrac{d^3p'}{(2\pi)^3p'^0} \dfrac{d^3q}{(2\pi)^3q^0} \dfrac{d^3q'}{(2\pi)^3q'^0} \, s \upvarsigma(s,\theta) \, \delta^4 (p{+}p'{-}q{-}q') \, f^{\text{eq}}_{p} f^{\text{eq}}_{p'} |\phi_p^{(n)}+\phi_{p'}^{(n)}-\phi_q^{(n)}-\phi_{q'}^{(n)}|^2  \, ,
\end{equation}
with equilibrium distribution $f^{\text{eq}}_p=e^{\alpha-\beta p^0}$, differential cross section $\upvarsigma(s,\theta) \leq A/s^a$ (where $A{\geq} 0$ and $0{\leq} a {\leq} 1$), and 
\begin{equation}\label{sequenceiazio}
    \phi_p^{(n)}= e^{\frac{\beta}{2}(1-n^{-1})p^0} \dfrac{2p^2 p^3}{(p^2)^2{+}(p^3)^2} \, . 
\end{equation}
Applying the Cauchy-Schwartz inequality to the dot product between $(1,1,-1,-1)$ and $(\phi_p^{(n)},\phi_{p'}^{(n)},\phi_q^{(n)},\phi_{q'}^{(n)})$, we obtain the following upper bound (which is an instance of AM-QM inequality):
\begin{equation}
  |\phi_p^{(n)}+\phi_{p'}^{(n)}-\phi_q^{(n)}-\phi_{q'}^{(n)}|^2 \leq 4 \bigg[|\phi_p^{(n)}|^2+|\phi_{p'}^{(n)}|^2+|\phi_q^{(n)}|^2+|\phi_{q'}^{(n)}|^2\bigg] 
\end{equation}
Therefore, recalling that $\upvarsigma(s,\theta) \leq A/s^a$, we have an upper bound on the entropy production rate,
\begin{equation}\label{veganism}
     \varsigma_n \leq  \dfrac{Ae^{2\alpha}}{2(2\pi)^{6}} 
\int \dfrac{d^3p}{ p^0} \dfrac{d^3p'}{p'^0} \dfrac{d^3q}{q^0} \dfrac{d^3q'}{q'^0} \, s^{1-a} \, \delta^4 (p{+}p'{-}q{-}q') \, e^{-\beta(p^0+p'^0)} \bigg[|\phi_p^{(n)}|^2+|\phi_{p'}^{(n)}|^2+|\phi_q^{(n)}|^2+|\phi_{q'}^{(n)}|^2\bigg] \, ,  
\end{equation}
which we can decompose into four parts:
\begin{equation}\label{vagn2}
\begin{split}
\varsigma_n \leq {}&  \dfrac{Ae^{2\alpha}}{2(2\pi)^{6}} 
\int \dfrac{d^3p}{ p^0} \dfrac{d^3p'}{p'^0} \dfrac{d^3q}{q^0} \dfrac{d^3q'}{q'^0} \, s^{1-a} \, \delta^4 (p{+}p'{-}q{-}q') \, e^{-\beta(p^0+p'^0)} |\phi_p^{(n)}|^2   \\
+ {}&  \dfrac{Ae^{2\alpha}}{2(2\pi)^{6}} 
\int \dfrac{d^3p}{ p^0} \dfrac{d^3p'}{p'^0} \dfrac{d^3q}{q^0} \dfrac{d^3q'}{q'^0} \, s^{1-a} \, \delta^4 (p{+}p'{-}q{-}q') \, e^{-\beta(p^0+p'^0)} |\phi_{p'}^{(n)}|^2   \\
+ {}&  \dfrac{Ae^{2\alpha}}{2(2\pi)^{6}} 
\int \dfrac{d^3p}{ p^0} \dfrac{d^3p'}{p'^0} \dfrac{d^3q}{q^0} \dfrac{d^3q'}{q'^0} \, s^{1-a} \, \delta^4 (p{+}p'{-}q{-}q') \, e^{-\beta(p^0+p'^0)} |\phi_q^{(n)}|^2   \\
+ {}&  \dfrac{Ae^{2\alpha}}{2(2\pi)^{6}} 
\int \dfrac{d^3p}{ p^0} \dfrac{d^3p'}{p'^0} \dfrac{d^3q}{q^0} \dfrac{d^3q'}{q'^0} \, s^{1-a} \, \delta^4 (p{+}p'{-}q{-}q') \, e^{-\beta(p^0+p'^0)} |\phi_{q'}^{(n)}|^2 \, .  \\
\end{split}
\end{equation}
In the second line, we perform the change on integration variables $(pp'qq')\rightarrow(p'pqq')$, finding that the first two lines are identical. Analogously, we perform the change of variables $(pp'qq')\rightarrow(pp'q'q)$ in the fourth line, finding that the last two lines are identical. Finally, we perform the change of variables $(pp'qq')\rightarrow(qq'pp')$ in the third line, and use the constraint $p^\mu+p'^\mu=q^\mu +q'^\mu$ imposed by the Dirac delta to show that the third line is identical to the first line\footnote{Note that $s=-(p^\mu {+}p'^\mu)(p_\mu {+}p'_\mu)=-(q^\mu {+}q'^\mu)(q_\mu {+}q'_\mu)$.}. Thus, all the lines are equal to each other, and we can write
\begin{equation}
    \varsigma_n \leq  \dfrac{2Ae^{2\alpha}}{(2\pi)^{6}} 
\int \dfrac{d^3p}{ p^0} \dfrac{d^3p'}{p'^0} \dfrac{d^3q}{q^0} \dfrac{d^3q'}{q'^0} \, s^{1-a} \, \delta^4 (p{+}p'{-}q{-}q') \, e^{-\beta(p^0+p'^0)} |\phi_p^{(n)}|^2  \, .
\end{equation}
Now, we bound the Mandelstam variable $s=-(p^\mu {+}p'^\mu)(p_\mu {+}p'_\mu)$ as follows\footnote{Applying the Cauchy-Schwartz inequality to the (Euclidean) dot product between $(m,\textbf{p})$ and $(m,-\textbf{p}')$, we find that $|m^2{-}p^j p'_j |\leq p^0 p'^0$.}:
   \begin{equation}
s=-p^\mu p_\mu -p'^\mu p'_\mu-2p^\mu p'_\mu =2m^2-2 p^j p'_j +2p^0 p'^0 \leq 4 p^0 p'^0 \, .
\end{equation}
Since $a$ does not exceed $1$, the quantity $s^{1-a}$ is a non-decreasing function of $s$, and we can set $s^{1-a}\leq (4 p^0 p'^0)^{1-a}$. We can also bound $|2p^2 p^3|$ with $(p^2)^2{+}(p^3)^2$ in \eqref{sequenceiazio}. Thus, we have
\begin{equation}\label{veganismo2}
\varsigma_n \leq  \dfrac{8 Ae^{2\alpha}}{4^a(2\pi)^{6}} 
\int \dfrac{d^3p}{ (p^0)^a} \dfrac{d^3p'}{(p'^0)^a} \dfrac{d^3q}{q^0} \dfrac{d^3q'}{q'^0} \,   \, \delta^4 (p{+}p'{-}q{-}q') \, e^{-\beta(p^0/n+p'^0)}  \, .  
\end{equation}
Defined $P^\mu =p^\mu +p'^\mu$, let us evaluate the integral
\begin{equation}
    K(P)= \int  \dfrac{d^3q}{q^0} \dfrac{d^3q'}{q'^0} \,   \, \delta^4 (P{-}q{-}q')  \, .
\end{equation}
Since the measures $d^3 q/q^0$ and $d^3 q/q^0$ are Lorentz-invariant, we can carry out the calculation in the center-of-momentum frame, where $P=(\sqrt{s},0,0,0)$. Then,
\begin{equation}
    K(P)= \int \dfrac{d^3q}{q^0} \dfrac{d^3q'}{q'^0} \,    \delta (\sqrt{s}{-}q^0{-}q'^0) \delta^3 (\textbf{q}+\textbf{q}') = \int  \dfrac{d^3q}{(q^0)^2}  \,    \delta (\sqrt{s}{-}2q^0) = 2\pi \sqrt{1-\dfrac{(2m)^2}{s}} \leq 2\pi \, .
\end{equation}
Thus, \eqref{veganismo2} becomes
\begin{equation}\label{produczionne}
\begin{split}
 \varsigma_n \leq {}&  \dfrac{8 Ae^{2\alpha}}{4^a(2\pi)^{5}} 
\int \dfrac{d^3p}{ (p^0)^a}   \, e^{-\beta p^0/n} \int \dfrac{d^3p'}{ (p'^0)^a}   \, e^{-\beta p'^0} \\ ={}& \dfrac{32 Ae^{2\alpha}}{4^a(2\pi)^3} 
 \int_m^{+\infty}  E^{1-a} e^{-\beta E/n} \sqrt{E^2{-}m^2} \, dE   \int_m^{+\infty}  E'^{1-a} e^{-\beta E'} \sqrt{E'^2{-}m^2} \, dE'   \, .\\
\end{split}
\end{equation}
The integrals can be bounded above replacing $m$ with $0$, and we finally obtain
\begin{equation}
 \varsigma_n \leq  \dfrac{32 Ae^{2\alpha} \Gamma(3{-}a)^2}{4^a(2\pi)^{3}\beta^{6-2a}} \,  n^{3-a} \, .
\end{equation}
On the other hand, the information density is given by
\begin{equation}
    E^0_n = \dfrac{1}{2} \int \dfrac{d^3 p}{(2\pi)^3} f^{\text{eq}}_p |\phi_p^{(n)}|^2 = \dfrac{e^\alpha}{(2\pi)^2} \int_0^{+\infty} dp \, p^2 \, e^{-\frac{\beta}{n} \sqrt{m^2+p^2}} \, .
\end{equation}
Since $\sqrt{m^2+p^2} \leq m+p$, we can bound below the information density as follows:
\begin{equation}\label{densitone}
  E^0_n \geq \dfrac{e^{\alpha - \frac{m\beta}{n}}}{(2\pi)^2} \int_0^{+\infty} dp \, p^2 \, e^{-\frac{\beta}{n} p} \geq \dfrac{2e^{\alpha - m\beta}}{(2\pi)^2 \beta^3} \, n^3 \, .   
\end{equation}
Taking the ratio between \eqref{produczionne} and \eqref{densitone}, we find that
\begin{equation}
    \dfrac{ \varsigma_n}{E_n^0} \leq \dfrac{Ae^{\alpha+m\beta} \Gamma(3{-}a)^2}{2^{2a-3} \pi\beta^{3-2a}} \,  n^{-a} \, ,
\end{equation}
which is what we wanted to prove.\\

\textbf{Remark:} While the function $\phi_p^{(\infty)}$ obtained as a limit of \eqref{sequenceiazio} does not belong to the Hilbert space $\mathcal{H}$ of the linear theory, it still constitutes (once multiplied by an infinitesimal parameter $\epsilon$) a sensible state in the non-linear regime. The corresponding distribution function is\footnote{A contribution of order $\epsilon^2$ is needed to guarantee that the full distribution function be non-negative at finite $\epsilon$. The interested reader can see section \ref{statemente} for more details.}
\begin{equation}
f_p^{(\infty)}= e^{\alpha-\beta p^0} + \epsilon \, e^{\alpha -\frac{\beta}{2}p^0} \dfrac{2p^2 p^3}{(p^2)^2{+}(p^3)^2} +\mathcal{O}(\epsilon^2) \, .
\end{equation}
As can be seen, the exponential factor $e^{\alpha -\frac{\beta}{2}p^0}$ in the excitation corresponds to a thermal state whose temperature is twice the temperature of the equilibrium state. This again shows that, in general, long-lived modes do not need to have a very large energy content. However, for an excitation to live longer than all soundwaves, $\phi_p$ should have unbounded support in $p$. 

\newpage

\section{Generalization of Theorem 2 to all soft interactions}

Here, we provide a generalization of Theorem 2, which is based on rigorous results about the essential spectrum of Boltzmann's collision operator in the case of soft interactions. Our main reference is \cite{Dudynski1988}.

\subsection{Duality of frameworks}\label{dualona}

First, let us compare our notation with that of \cite{Dudynski1988}. When constructing the relevant space of functions, \citet{Dudynski1988} do not use our linear degree of freedom $\phi_p$. Instead, their degree of freedom is\footnote{To see this, we note that, while we consider the decomposition $f_p=f^{\text{eq}}_p{+}f^{\text{eq}}_p \phi_p$, they adopt the decomposition $f_p=f^{\text{eq}}_p{+}\sqrt{f^{\text{eq}}_p} \,  \Phi_p$.} $\Phi_p=\sqrt{f^{\text{eq}}_p} \, \phi_p$, and the space of functions is $\tilde{\mathcal{H}}=L^2(\mathbb{R}^3,1)$, with inner product
\begin{equation}
    \langle \Psi_p,\Phi_p \rangle = \int \dfrac{d^3 p}{(2\pi)^3} \Psi^*_p \Phi_p \, . 
\end{equation}
It is immediate to see that, if $\Phi_p=\sqrt{f^{\text{eq}}_p} \, \phi_p$ and $\Psi_p=\sqrt{f^{\text{eq}}_p} \, \psi_p$, then the correpondening inner products coincide: 
\begin{equation}
    \langle \Psi_p,\Phi_p \rangle = \int \dfrac{d^3 p}{(2\pi)^3} \Psi^*_p \Phi_p = \int \dfrac{d^3 p}{(2\pi)^3} f^{\text{eq}}_p \psi^*_p \phi_p = (\psi_p,\phi_p)\, . 
\end{equation}
In other words, the mapping $\Phi_p=\sqrt{f^{\text{eq}}_p} \, \phi_p$ is an isometry between our Hilbert space $\mathcal{H}$ and the Hilbert space $\Tilde{\mathcal{H}}$ of \cite{Dudynski1988}. Furthermore, $\phi_p \in \mathcal{H}$ if and only if $\Phi_p \in \Tilde{\mathcal{H}}$. Indeed, it should be noted that \citet{Dudynski1988} defined $\tilde{\mathcal{H}}$ as the set of states with finite free energy\footnote{Actually, they use the expression ``finite entropy'', but the two are equivalent in the present setting \cite{GavassinoCausality2021}.}, which is also our definition for the space $\mathcal{H}$. 

Let us now focus on the collision operators. In the homogenous limit, the equation of motion considered by \cite{Dudynski1988} has the form $\partial_t \Phi_p=\mathcal{L}\Phi_p$. Comparing this with our equation of motion $\partial_t \phi_p=(p^0)^{-1}L\phi_p$, we see that their operator $\mathcal{L}$ and our operator $L$ are related by the identity (the right side should be interpreted as a composition of operators)
\begin{equation}
\mathcal{L}= \sqrt{f^{\text{eq}}_p} \, \dfrac{1}{p^0} L \, \dfrac{1}{\sqrt{f^{\text{eq}}_p}} \, .
\end{equation}
Hence, the matrix elements $\langle \Psi_p,\mathcal{L} \Phi_p \rangle$ in $\tilde{\mathcal{H}}$ coincide with the correponding matrix elements $\big( \psi_p,(p^0)^{-1}L\phi_p \big)$ in $\mathcal{H}$. The proof is straightforward:
\begin{equation}
\begin{split}
\langle \Psi_p,\mathcal{L}\Phi_p \rangle ={}& \int \dfrac{d^3 p}{(2\pi)^3} \Psi^*_p \mathcal{L} \Phi_p \\
={}& \int \dfrac{d^3 p}{(2\pi)^3} \sqrt{f^{\text{eq}}_p} \, \psi^*_p \sqrt{f^{\text{eq}}_p} \, \dfrac{1}{p^0} L \, \dfrac{1}{\sqrt{f^{\text{eq}}_p}} \sqrt{f^{\text{eq}}_p} \, \phi_p \\
={}& \int \dfrac{d^3 p}{(2\pi)^3} f^{\text{eq}}_p \, \psi^*_p \, \dfrac{1}{p^0} L  \phi_p = \bigg(\psi_p, \dfrac{1}{p^0} L \phi_p  \bigg) \, . \\
\end{split}
\end{equation}
Indeed, both $\langle \Phi_p,\mathcal{L}\Phi_p \rangle$ and $\big( \phi_p,(p^0)^{-1}L\phi_p \big)$ are the negative value of the entropy production rate. Finally, let us note that, called $||\Phi_p ||_{DE}^2 \equiv \langle \Phi_p, \Phi_p \rangle=(\phi_p,\phi_p)=||\phi_p||^2$ the Hilbert norm in $\Tilde{\mathcal{H}}$, we have the following identities:
\begin{equation}\label{fourtisixxxx}
||\mathcal{L}\Phi_p||^2_{DE}=\langle \mathcal{L}\Phi_p,\mathcal{L}\Phi_p \rangle =\bigg\langle \sqrt{f^{\text{eq}}_p} \, \dfrac{1}{p^0} L \phi_p, \sqrt{f^{\text{eq}}_p} \, \dfrac{1}{p^0} L \phi_p \bigg\rangle = \bigg( \dfrac{1}{p^0} L \phi_p,  \, \dfrac{1}{p^0} L \phi_p \bigg) =|| (p^0)^{-1}L\phi_p||^2 \, . 
\end{equation}
This also implies that $\mathcal{L}$ is defined on $\Phi_p$ if and only if $(p^0)^{-1}L$ is defined on the corresponding state $\phi_p$.

In conclusion, $\mathcal{H}$ and $\tilde{\mathcal{H}}$ are different representations of the same physical Hilbert space, just like, in quantum mechanics, one can represent the same physical wavefunction in position space or in momentum space, without affecting the value of the physical amplitudes. It follows that any result of \cite{Dudynski1988} that can be expressed solely in terms of inner products of physical states must be valid also within our representation.

\newpage
\subsection{Coincidence of the spectra}

Let us use the above duality to prove that the spectrum of $\mathcal{L}$ in $\tilde{\mathcal{H}}$ coincides with the spectrum of $(p^0)^{-1}L$ in $\mathcal{H}$. Since both $\mathcal{L}$ and $(p^0)^{-1}L$ are self-adjoint, in the respective spaces, all points in the spectrum are approximate eigenvalues, by Theorem 1 of \cite[Lecture 18]{Bhatiabook}. This means that a number $\lambda \in \mathbb{R}$ belongs to the spectrum of $\mathcal{L}$ if and only if there is a sequence of states $\Phi^{(n)}_p$ with non-vanishing norm such that
\begin{equation}
    \lim_{n\rightarrow \infty} \dfrac{||(\mathcal{L}-\lambda) \Phi_p^{(n)} ||_{DE}}{|| \Phi^{(n)}_p||_{DE}} =0 \, .
\end{equation}
On the other hand, the corresponding states $\phi_p^{(n)}$ are such that
\begin{equation}
    \lim_{n\rightarrow \infty} \dfrac{||(\mathcal{L}{-}\lambda) \Phi_p^{(n)} ||_{DE}}{|| \Phi^{(n)}_p||_{DE}} = \lim_{n\rightarrow \infty} \dfrac{||\big[(p^0)^{-1}L{-}\lambda\big] \phi_p^{(n)} ||}{|| \phi^{(n)}_p||}  \, ,
\end{equation}
by the duality established above, see, e.g., equation \eqref{fourtisixxxx}. It follows that $\lambda$ is an approximate eigenvalue of $\mathcal{L}$ if and only if it is an approximate eigenvalue of $(p^0)^{-1}L$, proving that the spectra are, indeed, the same.

\subsection{Soft interactions with finite mass}

In \cite{Dudynski1988}, it was shown that, if $m\neq 0$, then $\mathcal{L}$ is a bounded operator for all soft interactions, namely for all collision cross-sections $\upvarsigma(s,\theta)$ such that 
\begin{equation}\label{groupszz}
    \upvarsigma(s,\theta)< \dfrac{B\sin^c \theta}{(s{-}4m^2)^a} \, ,
\end{equation}
where $a,B,c$ are constants, with $0<a<2$, $B>0$, and $c>-2$. Given the duality discussed above, we conclude that also $(p^0)^{-1}L$ is bounded, with the same norm. Moreover, \citet{Dudynski1988} showed that the essential spectrum of $\mathcal{L}$ (again, for soft interactions with $m \neq 0$) always takes the form $[-D,0]$, for some constant $D>0$. Since the essential spectrum is a subset of the spectrum, and since the spectrum of $\mathcal{L}$ coincides with the spectrum of $(p^0)^{-1}L$, we conclude that the spectrum of $(p^0)^{-1}L$ contains the interval $[-D,0]$. Therefore, all soft interactions fulfill criterion (d) of our Theorem 1, and facts (a,b,c) automatically follow.

The above derivation provides an additional characterization of the spectrum of massive gases with soft interactions. In fact, strictly speaking, statement (d) of Theorem 1 only guarantees the existence of a sequence of eigenfrequencies that accumulate in 0. However, if \eqref{groupszz} holds, such eigenfrequencies form a continuum, since they cover the interval $\omega \in -i[0,D]$. Thus, in these systems, there is really a branch cut, at least for finite mass.

We emphasize that the distinction between a continuous cut and an infinite sequence of discrete poles $\omega_n \rightarrow 0$ is merely academic and impossible to detect experimentally. In fact, one would need infinite resolution in $\omega$ to measure the separation of all discrete poles $\omega_n$ in a neighborhood of 0. This would take infinite time, due to the uncertainty relation $\Delta t \, \Delta \omega \gtrsim 1$ in Fourier analysis.

\newpage

\section{Non-linear considerations}

All our calculations in this article were performed in the linear regime. In most physics discussions, it is taken for granted that the properties of the linear theory should reflect the non-linear behavior in proximity to the unperturbed state. Indeed, this is a safe assumption in most cases. However, in a problem like ours, where many different limits are considered, and the order of limits in general matters, the reliability of the linear theory is not unquestionable. Hence, it would be useful to have a rigorous proof that, for each of the linear states considered in the present article, there exists a corresponding family of non-linear states whose features are well approximated by the given linear state near equilibrium. Here, we provide such a proof.

\subsection{Statement of the problem}\label{statemente}

First, let us state precisely what needs to be proven. Given that our study concerns the properties of the linear operator $(p^0)^{-1}L$, and how it acts on certain predetermined states [like \eqref{exciuto}, \eqref{exciutone}, \eqref{phiF4f}, and \eqref{sequenceiazio}], we must show that, for any given linear state $\phi_p \in \mathcal{H}$ of interest (with $\phi_p \in \mathbb{R}$), it is possible to build a one-parameter family of non-linear states $f_p(\epsilon)$, with $f_p(0)=f^\text{eq}_p$, and such that the following two facts hold:
\begin{equation}\label{grinizo}
\begin{split}
\dfrac{f_p(\epsilon){-}f^{\text{eq}}_p}{\epsilon} \xrightarrow{\epsilon \rightarrow 0} {}&f^{\text{eq}}_p \phi_p \, , \\
\dfrac{1}{p^0}\dfrac{\mathcal{C}_p(\epsilon)}{\epsilon} \xrightarrow{\epsilon \rightarrow 0} {}& f^{\text{eq}}_p \dfrac{1}{p^0} L\phi_p \, , \\
\end{split}
\end{equation}
where $\mathcal{C}_p(\epsilon)$ is the non-linear collision integral evaluated on $f_p(\epsilon)$. Now, two important subtleties need to be mentioned. The first subtlety is that, since we want each $f_p(\epsilon)$ to represent a physically meaningful non-linear state, we need to make sure that $f_p(\epsilon) \geq 0$ (for all $p$), at least within a \textit{finite} neighborhood of $\epsilon=0$. This implies that, in general, we cannot consider the family $f_p(\epsilon)=f^{\text{eq}}_p(1{+}\epsilon \phi_p)$, because $1{+}\epsilon \phi_p$ may fail to be non-negative in any finite neighborhood of $\epsilon=0$. For example, take $\phi_p=-p^0 \in \mathcal{H}$. No matter how small $\epsilon>0$ is, the function $f_p =f_p^{\text{eq}}(1-\epsilon p^0)$ will always become negative for $p^0>1/\epsilon$. To avoid this inconvenience, we will consider the ``regularised'' family
\begin{equation}\label{regularizzzzzed}
f_p(\epsilon)=f^\text{eq}_p \bigg[1+\epsilon \phi_p \, e^{-(\epsilon \phi_p)^2}\bigg] \, ,
\end{equation}
which fulfills the following bounds (valid for any $\phi_p \in \mathbb{R}$ and for any $\epsilon \in \mathbb{R}$): 
\begin{equation}
\bigg(1{-}\dfrac{1}{\sqrt{2e}} \bigg) f^\text{eq}_p \leq f_p(\epsilon) \leq \bigg(1{+}\dfrac{1}{\sqrt{2e}} \bigg) f^\text{eq}_p \, . 
\end{equation}
This guarantees not only that $f_p(\epsilon)$ is always positive, but also that physical observables like the energy and the particle number are finite, no matter how large $\phi_p$ or $\epsilon$ is. Therefore, all states $f_p(\epsilon)$ are physically accessible. Additionally, we note that the energy content of such states is not particularly high, since the ratio $f_p(\epsilon)/f^\text{eq}_p$ is always of order 1.

The second subtle issue is how we should interpret the limits in equation \eqref{grinizo}. It may be tempting to require the limits to hold for each fixed $p$ (i.e. ``pointwise'' in $p$). However, pointwise convergence is not a good measure of distance, because it does not capture global features (such as integrals \cite[Ch 7, Example 7.6]{RudinAnalysis_book}). Instead, since we are mostly interested in the collective motion of the gas, and since $f_p$ (being a distribution) has physical significance only inside momentum integrals \cite[\S I.1]{Groot1980RelativisticKT}, it is more natural to interpret the limit in \eqref{grinizo} as a convergence in $L^1$ norm\footnote{Note that, if \eqref{grinizo} holds in $L^1$ norm, then we can find a discrete sequence of values $\epsilon_n \rightarrow 0$ along which \eqref{grinizo} holds also pointwise almost everywhere, and almost uniformly \cite[Corollary 1.5.10]{TaoMeasureTheory2011}.}, the latter being defined as follows:
\begin{equation}
||A_p||_{L^1} =\int \dfrac{d^3 p}{(2\pi)^3}  |A_p| \,  .
\end{equation}
This is the perfect measure of distance for our purposes, because it ``counts the particles''\footnote{Indeed, $||f_p||_{L^1}$ is just the number density.}. For example, if two distribution functions $f_p$ and $g_p$ are such that $|| f_p-g_p||_{L^1}=X$, then we may interpret $X$ as the number of particles (per unit volume) that we need to ``disturb'' if we want to move from $f_p$ to $g_p$. Therefore, if we manage to prove that property \eqref{grinizo} holds in $L^1$ norm, we will be entitled to say that, for small $\epsilon$, the non-linear theory differs from the linear theory by a negligible number of particles.

\subsection{Convergence of the state}
\vspace{-0.2cm}

Let us prove the first limit in \eqref{grinizo}, for the family of states \eqref{regularizzzzzed}. The function of interest is
\begin{equation}
\dfrac{f_p(\epsilon){-}f^{\text{eq}}_p}{\epsilon} = f^{\text{eq}}_p \phi_p \, e^{-(\epsilon \phi_p)^2} \, .
\end{equation}
Clearly, it converges to $f^{\text{eq}}_p\phi_p$ in a pointwise sense. Does the convergence hold also in $L^1$ norm? To check this, we need to calculate the following limit:
\begin{equation}\label{todominate}
\lim_{\epsilon \rightarrow 0} \bigg|\bigg| \dfrac{f_p(\epsilon)-f_p^{\text{eq}}}{\epsilon} - f^{\text{eq}}_p\phi_p \bigg| \bigg|_{L^1}  = \lim_{\epsilon \rightarrow 0} \int \dfrac{d^3 p}{(2\pi)^3}  f^{\text{eq}}_p |\phi_p| \bigg[1- e^{-(\epsilon \phi_p)^2}\bigg] \, .
\end{equation}
Now, let us note that, since $\phi_p \in \mathcal{H}$, we can apply the Cauchy-Schwartz inequality in $\mathcal{H}$ to show that 
\begin{equation}
\int \dfrac{d^3 p}{(2\pi)^3} f^{\text{eq}}_p |\phi_p| =\big(1,|\phi_p| \big) \leq ||1|| \times  ||\phi_p ||<\infty \, .
\end{equation}
Thus, we can use the dominated convergence theorem in \eqref{todominate} to move the limit inside the sign of integration (dominating the square bracket with 1). Then, we have
\begin{equation}
\lim_{\epsilon \rightarrow 0} \bigg|\bigg| \dfrac{f_p(\epsilon)-f_p^\text{eq}}{\epsilon} - f^{\text{eq}}_p\phi_p \bigg| \bigg|_{L^1}  = \int \dfrac{d^3 p}{(2\pi)^3}  f^{\text{eq}}_p |\phi_p| \bigg[1- \lim_{\epsilon \rightarrow 0} e^{-(\epsilon \phi_p)^2}\bigg] = \int \dfrac{d^3 p}{(2\pi)^3} f^{\text{eq}}_p |\phi_p| \bigg[1- 1\bigg]=0\, ,
\end{equation}
which is what we wanted to prove. This shows that, in the limit of small $\epsilon$, the non-linear distribution function $f_p(\epsilon)$ is described with arbitrarily good accuracy by the linear theory (through the linearised state $\phi_p$). With a similar argument as above, one can also prove that all the moments $\rho^{\mu \alpha_1 ... \alpha_n}(\epsilon)$ are well approximated by the linear theory.

\vspace{-0.2cm}
\subsection{Convergence of the collision integral}
\vspace{-0.2cm}

Let us now discuss the second limit in \eqref{grinizo}, which is a bit harder. Invoking the bilinear structure of Boltzmann's integral and the definition of $L$ \cite[\S 9.3.1]{cercignani_book}, it is straightforward to show that
\begin{equation}
\dfrac{\mathcal{C}_p(\epsilon)}{\epsilon} = f_p^\text{eq}L [\phi_p e^{-(\epsilon\phi_p)^2}]+\epsilon \, \mathcal{C}[f_p^{\text{eq}}\phi_p e^{-(\epsilon\phi_p)^2}] \, .
\end{equation}
Hence, we have the following relations:
\begin{equation}\label{qqurello}
\begin{split}
\lim_{\epsilon \rightarrow 0} \bigg|\bigg| \dfrac{1}{p^0} \dfrac{\mathcal{C}(\epsilon)}{\epsilon} - f^{\text{eq}}_p \dfrac{1}{p^0} L\phi_p \bigg| \bigg|_{L^1}  ={}& \lim_{\epsilon \rightarrow 0} \bigg|\bigg| -f_p^\text{eq} \dfrac{1}{p^0}L [\phi_p(1- e^{-(\epsilon\phi_p)^2})]+\epsilon \,  \dfrac{1}{p^0}\mathcal{C}[f_p^{\text{eq}}\phi_p e^{-(\epsilon\phi_p)^2}] \bigg| \bigg|_{L^1} \\
\leq {}& \lim_{\epsilon \rightarrow 0} \bigg|\bigg| f_p^\text{eq} \dfrac{1}{p^0} L [\phi_p(1- e^{-(\epsilon\phi_p)^2})] \bigg| \bigg|_{L^1}+ \lim_{\epsilon \rightarrow 0} \epsilon \bigg| \bigg| \dfrac{1}{p^0} \mathcal{C}[f_p^{\text{eq}}\phi_p e^{-(\epsilon\phi_p)^2}] \bigg| \bigg|_{L^1} \, ,
\end{split}
\end{equation}
where we used the triangle inequality of $L^1$. Let us study each of the two limits on the second line separately.

\subsubsection{The first limit}

Introducing the short-hand notation $\Bar{\phi}_p =\phi_p(1- e^{-(\epsilon\phi_p)^2})$,
we have the following chain of relations:
\begin{equation}\label{fiursto}
\begin{split}
\bigg|\bigg| f_p^\text{eq} \dfrac{1}{p^0}L\Bar{\phi}_p \bigg| \bigg|_{L^1}={}& \int \dfrac{d^3 p}{(2\pi)^3 p^0} f^\text{eq}_p \bigg| \int \dfrac{d^3p'}{(2\pi)^3p'^0} \dfrac{d^3q}{(2\pi)^3q^0} \dfrac{d^3q'}{(2\pi)^3q'^0} W_{pp'\leftrightarrow qq'} f^{\text{eq}}_{p'}(\Bar{\phi}_q{+}\Bar{\phi}_{q'}{-}\Bar{\phi}_p{-}\Bar{\phi}_{p'}) \bigg| \\  
\leq {}& \int \dfrac{d^3 p}{(2\pi)^3 p^0}   \dfrac{d^3p'}{(2\pi)^3p'^0} \dfrac{d^3q}{(2\pi)^3q^0} \dfrac{d^3q'}{(2\pi)^3q'^0} W_{pp'\leftrightarrow qq'} f^\text{eq}_p f^{\text{eq}}_{p'}(|\Bar{\phi}_q|{+}|\Bar{\phi}_{q'}|{+}|\Bar{\phi}_p|{+}|\Bar{\phi}_{p'}|)  \\ 
= {}& 4\int \dfrac{d^3 p}{(2\pi)^3 p^0}   \dfrac{d^3p'}{(2\pi)^3p'^0} \dfrac{d^3q}{(2\pi)^3q^0} \dfrac{d^3q'}{(2\pi)^3q'^0} W_{pp'\leftrightarrow qq'} f^\text{eq}_p f^{\text{eq}}_{p'}|\Bar{\phi}_q| \\ 
= {}& 4\int \dfrac{d^3 p}{(2\pi)^3}   \dfrac{d^3p'}{(2\pi)^3} \dfrac{F}{p^0p'^0} \sigma_{\text{tot}}(s)  f^\text{eq}_p f^{\text{eq}}_{p'}|\Bar{\phi}_q| \\
= {}& 4\int \dfrac{d^3 p}{(2\pi)^3 }     f^\text{eq}_p|\nu_p\Bar{\phi}_q| =4(1,|\nu_p\Bar{\phi}_p|) \, . \\
\end{split} 
\end{equation}
In the first line, we used the definition of $L$. In the second line, we used the fact that the absolute value of a sum cannot exceed the sum of the absolute values. In the third line, we performed changes or variables analogous to those that lead from \eqref{veganism} to \eqref{vagn2}. In the fourth line, we invoked the definition of invariant flux $F$ and total cross-section $\sigma_{\text{tot}}(s)$, see \cite[\S I.2.d]{Groot1980RelativisticKT}. In the last line, we invoked the definition of collisional frequency $\nu_p \geq 0$ \cite{Strain2010}.

Now, when $\epsilon \rightarrow 0$, the norm of $\phi_p(1- e^{-(\epsilon\phi_p)^2})$ converges to 0. Again, this is because $|\phi_p|$ belongs to $\mathcal{H}$, and thus we can apply the dominated convergence theorem (dominating the round bracket with 1). By contrast, the function $|\nu_p \phi_p|$ may not have a finite norm in general. However, we recall that our Theorem 1 (and thus all our analysis), is concerned with states $\phi_p$ such that $(p^0)^{-1}L\phi_p$ belongs to $\mathcal{H}$. Using the correspondence with \cite{Dudynski1988}, it is straightforward to show that this requires that the norm of $|\nu_p\phi_p|$ be finite. Thus, we can again apply the dominated convergence theorem, and conclude that also the norm of $\nu_p \phi_p(1- e^{-(\epsilon\phi_p)^2})$ converges to 0 at small $\epsilon$. Therefore, we have that
\begin{equation}\label{rngomk}
\lim_{\epsilon \rightarrow 0} \bigg|\bigg| f_p^\text{eq} \dfrac{1}{p^0}L [\phi_p(1- e^{-(\epsilon\phi_p)^2})] \bigg| \bigg|_{L^1} \leq 4 \lim_{\epsilon \rightarrow 0} \bigg(1,|\nu_p\phi_p(1-e^{-(\epsilon\phi_p)^2})| \bigg) =0 \, .
\end{equation}

\subsubsection{The second limit}

We have the following chain of relations:
\begin{equation}
\begin{split}
\bigg| \bigg| \dfrac{1}{p^0} \mathcal{C}[f_p^{\text{eq}}\phi_p e^{-(\epsilon\phi_p)^2}] \bigg| \bigg|_{L^1}{=}{}& \int \dfrac{d^3 p}{(2\pi)^3p^0}  \bigg|
 \int \dfrac{d^3p'}{(2\pi)^3p'^0} \dfrac{d^3q}{(2\pi)^3q^0} \dfrac{d^3q'}{(2\pi)^3q'^0} W_{pp'\leftrightarrow qq'} f^{\text{eq}}_{p} f^{\text{eq}}_{p'}(\phi_q\phi_{q'} e^{-\epsilon^2 (\phi_q^2+\phi_{q'}^2)}{-}\phi_p\phi_{p'} e^{-\epsilon^2 (\phi_p^2+\phi_{p'}^2)}) \bigg| \\
 {\leq} {}& \int \dfrac{d^3 p}{(2\pi)^3 p^0}\dfrac{d^3p'}{(2\pi)^3p'^0} \dfrac{d^3q}{(2\pi)^3q^0} \dfrac{d^3q'}{(2\pi)^3q'^0} W_{pp'\leftrightarrow qq'}  f^{\text{eq}}_{p} f^{\text{eq}}_{p'}(|\phi_q\phi_{q'}| e^{-\epsilon^2 (\phi_q^2+\phi_{q'}^2)}{+}|\phi_p\phi_{p'}| e^{-\epsilon^2 (\phi_p^2+\phi_{p'}^2)})  \\
 {\leq} {}& \int \dfrac{d^3 p}{(2\pi)^3 p^0}\dfrac{d^3p'}{(2\pi)^3p'^0} \dfrac{d^3q}{(2\pi)^3q^0} \dfrac{d^3q'}{(2\pi)^3q'^0} W_{pp'\leftrightarrow qq'}  f^{\text{eq}}_{p} f^{\text{eq}}_{p'}(|\phi_q\phi_{q'}|{+}|\phi_p\phi_{p'}| )  \\
\leq {}& \dfrac{1}{2} \int \dfrac{d^3 p}{(2\pi)^3 p^0}\dfrac{d^3p'}{(2\pi)^3p'^0} \dfrac{d^3q}{(2\pi)^3q^0} \dfrac{d^3q'}{(2\pi)^3q'^0} W_{pp'\leftrightarrow qq'}  f^{\text{eq}}_{p} f^{\text{eq}}_{p'}(|\phi_q|^2 {+}|\phi_{q'}|^2{+}|\phi_p|^2{+}|\phi_{p'}|^2 )  \\
= {}& 2 \int \dfrac{d^3 p}{(2\pi)^3 p^0}\dfrac{d^3p'}{(2\pi)^3p'^0} \dfrac{d^3q}{(2\pi)^3q^0} \dfrac{d^3q'}{(2\pi)^3q'^0} W_{pp'\leftrightarrow qq'} f^{\text{eq}}_{p} f^{\text{eq}}_{p'}|\phi_p|^2  \\
={}&  2   \int \dfrac{d^3 p}{(2\pi)^3}\dfrac{d^3p'}{(2\pi)^3} \dfrac{F}{p^0p'^0} \sigma_{\text{tot}}(s)  f^{\text{eq}}_{p} f^{\text{eq}}_{p'}|\phi_p|^2  \\
={}&  2   \int \dfrac{d^3 p}{(2\pi)^3}f^{\text{eq}}_{p} \nu_p |\phi_p|^2 =2(\phi_p,\nu_p \phi_p)\, . \\
\end{split}  
\end{equation}
In the first line, we wrote $\mathcal{C}$ explicitly. In the second line, we used the fact that the absolute value of a sum cannot exceed the sum of the absolute values. In the third line, we use the fact that $e^{-\epsilon^2 \phi^2_p}\leq 1$. In the fourth line, we used the well-known identity $2|ab|\leq |a|^2 +|b|^2$. In the remaining lines, we followed analogous steps to those in \eqref{fiursto}.

Again, let us recall that, in the cases we are interested in, both $\phi_p$ and $\nu_p \phi_p$ belong to $\mathcal{H}$, and hence $(\phi_p,\nu_p\phi_p)$ is finite. Then, we can write
\begin{equation}\label{unstoppable}
\lim_{\epsilon \rightarrow 0} \epsilon \bigg| \bigg| \dfrac{1}{p^0} \mathcal{C}[f_p^{\text{eq}}\phi_p e^{-(\epsilon\phi_p)^2}] \bigg| \bigg|_{L^1} \leq 2(\phi_p,\nu_p\phi_p) \lim_{\epsilon \rightarrow 0} \epsilon  =0 \, . 
\end{equation}

\subsubsection{Putting everything together}

Combining \eqref{qqurello}, \eqref{rngomk}, and \eqref{unstoppable}, we finally conclude that
\begin{equation}
\lim_{\epsilon \rightarrow 0} \bigg|\bigg| \dfrac{1}{p^0} \dfrac{\mathcal{C}(\epsilon)}{\epsilon} - f^{\text{eq}}_p \dfrac{1}{p^0} L\phi_p \bigg| \bigg|_{L^1}  = 0 \, ,   
\end{equation}
which is what we wanted to prove. This shows that, in the limit of small $\epsilon$, the action of $\mathcal{C}$ on the non-linear state $f_p(\epsilon)$ is approximated by the action of $L$ on $\phi_p$ to arbitrarily good accuracy.

\label{lastpage}

\end{document}